\documentclass[letterpaper,twocolumn,10pt]{article}

\usepackage[utf8]{inputenc}
\usepackage{url}
\usepackage{hyperref}
\usepackage{amsfonts,amssymb,amsmath,amsthm,mathtools}
\usepackage{float}
\usepackage{xspace}
\usepackage{xcolor}
\usepackage[justification=justified,format=plain]{caption} 
\usepackage{enumitem}
\usepackage{paralist}
\usepackage[capitalize]{cleveref}
\usepackage{subcaption}
\usepackage{graphicx}
\usepackage[boxruled]{algorithm2e}
\usepackage{bytefield}
\usepackage[normalem]{ulem}
\usepackage{booktabs}
\usepackage{setspace}
\usepackage[notquote]{hanging}
\usepackage{authblk}
\usepackage{fullpage}
\usepackage{multicol}

\theoremstyle{definition}
\newtheorem{theorem}{Theorem}
\newtheorem{definition}[theorem]{Definition}

\makeatletter
\newcommand\footnoteref[1]{\protected@xdef\@thefnmark{\ref{#1}}\@footnotemark}
\makeatother

\makeatletter
\def\blfootnote{\xdef\@thefnmark{}\@footnotetext}
\makeatother

\long\def\com#1{}

\newcommand{\ie}{{\em i.e.,}\xspace}
\newcommand{\eg}{{\em e.g.,}\xspace}
\newcommand{\st}{{s.t.}\xspace}

\newcommand{\N}{\mathbb N}

\newcommand{\G}{\mathbb G}
\newcommand{\E}{\ensuremath{\mathcal{E}}\xspace}
\newcommand{\D}{\ensuremath{\mathcal{D}}\xspace}
\newcommand{\M}{\ensuremath{\mathcal{M}}\xspace}
\newcommand{\V}{\ensuremath{\mathcal{V}}\xspace}
\newcommand{\hasha}{\textnormal{H}\xspace}
\newcommand{\hashb}{\ensuremath{\hat{\textnormal{H}}}\xspace}
\newcommand{\hashpayload}{\ensuremath{\textnormal{H}'}\xspace}
\newcommand{\mac}{\textsf{MAC}\xspace}

\newcommand{\A}{\ensuremath{\mathcal{A}}\xspace} 
\newcommand{\B}{\ensuremath{\mathcal{B}}\xspace} 

\newcommand{\cpa}{\textnormal{ind\$-cpa}\xspace}

\newcommand{\ccatwo}{\textnormal{ind\$-cca2}\xspace}

\newcommand{\indcca}{\textnormal{ind-cca2}\xspace}
\newcommand{\suite}{\ensuremath{S}\xspace}
\newcommand{\onelambda}{1^\lambda}
\newcommand{\getrand}{\overset{\$}{\gets}}

\newcommand{\meta}{\textsf{meta}\xspace}
\newcommand{\keys}{\textsf{keys}\xspace}
\newcommand{\aux}{\textsf{aux}\xspace}
\newcommand{\ctxtpayload}{c_{\textsf{payload}}}

\newcommand{\mackey}{\ensuremath{K_{mac}}\xspace}
\newcommand{\enckey}{\ensuremath{K_{enc}}\xspace}

\newcommand{\hide}{\texttt{Hide}\xspace}
\newcommand{\unhide}{\texttt{Unhide}\xspace}
\newcommand{\setup}{\textsf{Setup}\xspace}
\newcommand{\keygen}{\textsf{KeyGen}\xspace}
\newcommand{\enc}{\textsf{Enc}\xspace}
\newcommand{\dec}{\textsf{Dec}\xspace}
\newcommand{\qdec}{\textsf{qDec}}

\newcommand{\encap}{\textsf{Encap}\xspace}
\newcommand{\decap}{\textsf{Decap}\xspace}
\newcommand{\qdecap}{\textsf{qDecap}}
\newcommand{\kem}{\textsf{KEM.}}
\newcommand{\ies}{\textsf{IES.}}
\newcommand{\msbe}{\textsf{MSBE.}}
\newcommand{\hdr}{\textsf{HdrPURB.}}

\newcommand{\psetup}{\textsf{MsPURB.Setup}\xspace}
\newcommand{\pkeygen}{\textsf{MsPURB.KeyGen}\xspace}
\newcommand{\penc}{\textsf{MsPURB.Enc}\xspace}
\newcommand{\pdec}{\textsf{MsPURB.Dec}\xspace}

\newcommand{\headpurb}{\textsc{HdrPURB}\xspace}
\newcommand{\layout}{\textsc{Layout}\xspace}

\newtheoremstyle{block}
	{0.5em}
	{}
	{}
	{}
	{\scshape}
	{.}
	{\newline}
	{\underline{\thmname{#1}\thmnote{ #3}}}
\theoremstyle{block}
\newtheorem*{syntax}{Syntax}
\newtheorem*{game}{Game}
\newtheorem*{insta}{Instantiation}
\newtheorem*{algo}{Algorithms}

\newtheorem*{indist}{}

\newcommand{\gamestart}{\ensuremath{G_0}\xspace}
\newcommand{\gamememkeys}{\ensuremath{G_1}\xspace}
\newcommand{\gameh}[1]{\ensuremath{H_{#1}}\xspace}
\newcommand{\gamesubkeys}{\ensuremath{G_2}\xspace}
\newcommand{\gamememmsg}{\ensuremath{G_3}\xspace}
\newcommand{\gamesubenc}{\ensuremath{G_4}\xspace}
\newcommand{\gamerepbot}{\ensuremath{G_5}\xspace}
\newcommand{\gamesubmac}{\ensuremath{G_6}\xspace}
\newcommand{\gamesubpay}{\ensuremath{G_7}\xspace}

\newcommand{\gameinstart}{\ensuremath{G_{0}}\xspace}
\newcommand{\gamerandkey}{\ensuremath{G_{1}}\xspace}
\newcommand{\gameinend}{\ensuremath{G_{2}}\xspace}

\DeclarePairedDelimiter{\ceil}{\lceil}{\rceil}
\DeclarePairedDelimiter{\floor}{\lfloor}{\rfloor}

\newcommand{\hangone}{\hangindent=1em \hangafter=1 \noindent}
\newcommand{\hangtwo}{\hangindent=1em \hangafter=2 \noindent}

\newcommand{\para}[1]{\vspace*{0.4em} \noindent \textbf{\mbox{#1}}}
\newcommand{\summary}[1]{\noindent{\leavevmode\color{blue}[#1]}}
\newcommand{\personComment}[2]{{\noindent\color{red}[\underline{#1:} #2]}}

\newcommand{\kirill}[1]{\personComment{KN}{#1}}

\newif\ifAnnotations
\Annotationstrue

\ifAnnotations
\else
  \renewcommand{\summary}[1]{}
\fi

\newcommand{\padname}{\textsc{Padmé}\xspace}
\newcommand{\padme}{\textsc{Padmé}\xspace}
\newcommand{\npot}{\textsc{NextP2}\xspace}
\newcommand{\purb}{\textsc{PURB}\xspace}
\newcommand{\purbs}{\textsc{PURB}s\xspace}
\newcommand{\mspurb}{\textsc{MsPURB}\xspace}

\newcommand{\purbsurl}{\href{https://github.com/dedis/purb}{https://github.com/dedis/purb}}

\begin{document}
	
	\author[ ]{Kirill Nikitin{$^{*\dagger}$}}
	
	\author[ ]{Ludovic Barman{$^{*\dagger}$}}
	
	\author[ ]{Wouter Lueks{$^{\dagger}$}}
	
	\author[ \hspace{-0.6ex}]{Matthew Underwood}
	
	\author[ ]{Jean-Pierre Hubaux{$^{\dagger}$}}
	
	\author[ ]{Bryan Ford{$^{\dagger}$}}
	
	\affil[ ]{$^{\dagger}$\'{E}cole polytechnique f\'{e}d\'{e}rale de Lausanne, 
		Switzerland\vspace{2pt}}
	\affil[ ]{{\normalsize \texttt{firstname.lastname@epfl.ch}}} 

\title{\vspace{-1em}Reducing Metadata Leakage from Encrypted Files and 
	Communication with PURBs}

\date{}

\maketitle

\blfootnote{\llap{*}Share first authorship.}

\begin{abstract}
{
Most encrypted data formats
leak metadata via their plaintext headers,
such as format version, encryption schemes used,
number of recipients who can decrypt the~data,
and even the recipients' identities.
This leakage can pose security and privacy risks to users,
\eg by revealing the full membership of a group of collaborators
from a single encrypted e-mail,
or by enabling an eavesdropper to fingerprint the~precise 
encryption software version and configuration the~sender used.
\\
We propose that future encrypted data formats
improve security and privacy hygiene by producing
{\em Padded Uniform Random Blobs} or \purbs:
ciphertexts indistinguishable from random bit strings
to anyone without a decryption key.
A \purb's content leaks {\em nothing at all},
even the application that created it,
and is padded such that even its length
leaks as little as possible.
\\
Encoding and decoding ciphertexts with {\em no} cleartext markers
presents efficiency challenges, however.
We present cryptographically agile encodings enabling
legitimate recipients to decrypt a \purb efficiently,
even when encrypted for any number of recipients' public keys
and/or passwords,
and when these public keys are from different cryptographic suites.
\purbs employ \padme, a~novel padding scheme that limits
information leakage via ciphertexts of maximum length~$M$
to a practical optimum of $O(\log \log M)$ bits,
comparable to padding to a power of two,
but with lower overhead
of at most $12\%$ and decreasing with larger payloads.
}
\com{
Most encryption schemes, whether designed for communication or file 
storage, leak various metadata information in their plaintext headers,
either justifying it with efficiency reasons or leaving it out of the scope.
This leakage poses serious security and privacy risks, as it can 
be leveraged by an adversary to infer information about the communication
content or patterns, prove the use of a certain software, or confirm the 
existence of this communication, which can be incriminating in itself.
This paper presents an encoding and padding schemes for producing 
\purbs, \emph{Padded Uniform Random Blobs}. 
A \purb is an encrypted blob 
 -- corresponding to an application-level unit of data, be it a file or network 
message -- that is indistinguishable from a uniform string of random bits to 
any observer without the appropriate keys.
Despite not having plaintext metadata, 
a \purb can be efficiently decoded by the legitimate recipient(s), 
and a single \purb can support multiple cryptographic schemes and keys at 
the same time.
The padding scheme \padname operates by restricting the size of the 
mantissa of the padded blob length. \padname has modest overhead, 
max $+12\%$ and decreasing with a file size, and only leaks 
$O(\log_2(\log_2L))$ bits of information about the plaintext length.
}

\end{abstract}

\section{Introduction}
\label{sec:intro}

Traditional encryption schemes and protocols aim to 
protect only their data payload, leaving related metadata exposed.
Formats such as PGP~\cite{zimmermann95pgp} reveal
in cleartext headers
the public keys of the intended recipients,
the algorithm used for encryption,
and the actual length of the payload. 
Secure-communication protocols similarly leak information
during key and algorithm agreement.
The TLS handshake~\cite{rfc8446}, for example,
leaks in cleartext the protocol version, chosen cipher suite, and 
the public keys of the parties.
This metadata exposure is traditionally assumed not to be security-sensitive,
but important for the recipient's decryption efficiency.

Research has consistently shown, however,
that attackers can exploit metadata
to infer sensitive information about communication content.
In particular, an attacker may be able to
fingerprint users~\cite{pang07finger80211,valverde15bad} and 
the applications they use use~\cite{zhang11inferring}.
Using traffic analysis~\cite{danezis07introducing}, an attacker may be 
able to infer websites a user visited~\cite{danezis07introducing, 
panchenko11website, dyer12peek, wang13improved,wang16realistically} 
or videos a user watched~\cite{reed2016leaky, 
schuster2017beauty, reed2017identifying}.
On VoIP, metadata can be used to infer the 
geo-location~\cite{leblondskype}, the spoken 
language~\cite{wright2007language}, or the voice activity of 
users~\cite{chang2008inferring}.
Side-channel leaks from data compression~\cite{kelsey02compression} 
facilitate several attacks on SSL~\cite{rizzo13crime, gluck13breach, 
beery13time}. The lack of proper padding might enable an active attacker 
to learn the length of the user's password from 
TLS~\cite{vranken15bicycle} or QUIC~\cite{ringroad} traffic.
In social networks, metadata can be used to draw conclusions about users' 
actions~\cite{greschbach12devil}, whereas telephone metadata has been 
shown to be sufficient for user re-identification and for determining home 
locations~\cite{mayer16evaluating}.
Furthermore, by observing the format of packets, oppressive regimes 
can infer which technology is used and use this information for 
the purposes of incrimination or censorship.
Most TCP packets that Tor sends, for example,
are 586 bytes due to its standard cell 
size~\cite{herrmann2009website}.

As a step towards countering these privacy threats,
we propose that encrypted data formats should produce 
\emph{Padded Uniform Random Blobs} or \purbs:
ciphertexts designed to protect {\em all} encryption metadata.
A \purb 
encrypts application content and metadata into a single blob 
that is indistinguishable from a random string,
and is padded to minimize information leakage via its length
while minimizing space overhead.
Unlike traditional formats, a \purb does not leak 
the encryption schemes used, who or how many recipients can decrypt it,
or what application or software version created it.
While simple in concept,
because \purbs by definition contain {\em no} cleartext structure or markers,
encoding and decoding them efficiently presents practical challenges.

This paper's first key contribution is {\em Multi-Suite PURB} or \mspurb,
a cryptographically agile \purb encoding scheme that supports
any number of recipients, who can use either shared passwords or 
public-private key pairs utilizing multiple cryptographic suites. 
The main technical challenge is providing \emph{efficient} 
decryption to recipients without leaving any cleartext markers.
If efficiency was of no concern,
the sender could simply discard all metadata and expect the recipient
to parse and trial-decrypt the payload
using every possible format version, structure, and cipher suite.
Real-world adoption requires both decryption efficiency
and cryptographic agility, however.
\mspurb combines a variable-length header
containing encrypted metadata with a 
symmetrically-encrypted payload. 
The header's structure enables efficient decoding by legitimate 
recipients via a small number of trial decryptions.
\mspurb facilitates
the seamless addition and removal of supported cipher suites,
while leaking no information to third parties without a decryption key.
We construct our scheme starting with the standard construction of the 
Integrated Encryption Scheme (IES)~\cite{abdalla01oracle} and use the ideas of 
multi-recipient public-key 
encryption~\cite{kurosawa2002multi,bellare07multi} as a part of the 
multi-recipient development.

To reduce information leakage from data lengths,
this paper's second main contribution is
\padname, a padding scheme that
groups encrypted \purbs into indistinguishability sets
whose visible lengths
are representable as limited-precision floating-point numbers.
Like obvious alternatives such as padding to the next power of two,
\padme reduces maximum information leakage to $O(\log \log M)$ bits,
where $M$ is the maximum length of encrypted blob
a user or application produces.
\padme greatly reduces constant-factor overhead
with respect to obvious alternatives, however,
enlarging files by at most +$12\%$,
and less as file size increases.

In our evaluation, creating a~\mspurb ciphertext takes 
$235$\,ms for $100$ recipients on consumer-grade hardware
using $10$ different cipher suites,
and takes only $8$\,ms for the common single-recipient single-suite 
scenario. Our implementation is in pure Go without assembly 
optimizations that might speed up public-key operations.
Because the \mspurb design limits the number of costly public-key
operations, however,
decoding performance is comparable to PGP,
and is almost independent of the number of 
recipients (up to 10,000).

Analysis of real-world data sets show that
many objects are trivially identifiable by their unique sizes without padding,
or even after padding to a fixed block size
(\eg that of a block cipher or a Tor cell).
We show that \padname can significantly reduce the number of objects
uniquely identifiable by their sizes: from 83\%
to 3\% for 56k Ubuntu packages, from 87\% to 3\% for 191k 
Youtube videos, from 45\% to 8\% for 848k hard-drive user files, and from 
68\% to 6\% for 2.8k websites from the Alexa top 1M list.
This much stronger leakage protection
incurs an average space overhead of only 3\%.

\smallskip
\noindent In summary, our main contributions are as follows:
\begin{compactitem}
\item We introduce \mspurb, a novel encrypted data format
	that reveals no metadata information to observers without
	decryption keys, while efficiently supporting multiple recipients and 
	cipher suites.
\item We introduce \padname, a padding scheme that
	asymptotically minimizes information leakage from data lengths
	while also limiting size overheads.
\item We implement these encoding and padding schemes, 
	evaluating the former's performance against PGP
	and the latter's efficiency on real-world data.
\end{compactitem}

\section{Motivation and Background}
\label{sec:bg}

We first offer example scenarios in which \purbs may be useful,
and summarize the Integrated Encryption Scheme that we later use as 
a design starting point.

\subsection{Motivation and Applications}

\com{
Nowadays, various applications use encryption to provide data 
confidentiality to the users. Most of these applications leave some or all 
encryption metadata in cleartext, as it is often not considered to be 
of high risk to security or privacy. Nonetheless, some metadata types have 
already been shown as 
privacy-sensitive~\cite{danezis07introducing,panchenko11website,dyer12peek,
 wang13improved,wang16realistically, zhang11inferring, leblondskype, 
wright2007language, chang2008inferring, mayer16evaluating, 
greschbach12devil, schuster2017beauty}.  
We argue that a number of other metadata types, 
especially related to encryption, can also be potentially used to infer 
communication content or to decide on the attack vector.
Revealing what application has created the encrypted message can 
already be sensitive, as an attacker might censor or collect the traffic of this 
application specifically. The exposed encryption scheme or the 
application version might suggest to the~attacker to exploit 
implementation or cryptographic weaknesses, \eg RC4-encrypted traffic in 
TLS~\cite{rfc5246} has been discovered vulnerable to the ciphertext-only 
attack~\cite{alfardan13rc4}.
Finally, revealing the identities and/or the number of message recipients 
can enable an attacker to infer the group membership, \eg of a minority or 
activists group, or even enable rubber-hose decryption attacks. }

Our goal is to define a generic method applicable to most of 
the common data-encryption scenarios such that the techniques are 
flexible to the application type, to the cryptographic 
algorithms used, and to the number of participants involved. 
We also seek to enhance plausible 
deniability such that a user can deny that a \purb is created by a given 
application or that the user owns the key to decrypt it.
We envision several immediate applications that could benefit from using 
\purbs.

\para{E-mail Protection.}
E-mail systems traditionally use PGP or S/MIME for encryption.
Their packet formats~\cite{rfc4880}, however,
exposes format version, encryption 
methods, number and public-key identities of the recipients, and
public-key algorithms used. In addition, the payload is
padded only to the block size of a symmetric-key algorithm used,
which does not provide ``size privacy'',
as we show in \S\ref{sec:padme-eval}.
Using \purbs for encrypted e-mail could minimize this metadata leakage.
Furthermore, as e-mail 
traffic is normally sparse, the moderate overhead \purbs incur can easily be
accommodated.

\para{Initiation of Cryptographic Protocols.}
In most cryptographic protocols, initial cipher suite negotiation, 
handshaking, and key exchange are normally performed unencrypted. In 
TLS\,1.2~\cite{rfc5246}, an eavesdropper who monitors a connection from 
the start can learn many details such as cryptographic schemes used.
The unencrypted Server Name Indication (SNI)
enables an eavesdropper to determine which 
specific web site a client is connected to
among the sites hosted by the same server.
The eavesdropper can also fingerprint the 
client~\cite{ristic09http} or distinguish censorship-circumvention tools 
that try to mimic TLS traffic~\cite{houmansadr2013parrot,frolov19use}.
TLS\,1.3~\cite{rfc8446} takes a few protective measures:
\eg less unencrypted metadata during the handshake,
and an experimental 
extension for encrypted SNI~\cite{ritter14protecting,rfc8446}. 
These measures are only partial, however,
and leave other metadata, such as protocol version number, cipher suites,
and public keys, still visible.
\purbs could facilitate fully-encrypted handshaking from the start,
provided a client already knows at least one public key and 
cipher suite the server supports.
Clients might cache this information from prior connections,
or obtain it out-of-band while finding the server,
\eg via DNS-based authentication~\cite{rfc6698}.

\para{Encrypted Disk Volumes.}
VeraCrypt~\cite{veracrypt}
uses a block cipher to turn a disk partition into an 
\emph{encrypted volume} where the partition's free space is filled with 
random bits.
For plausible deniability and coercion protection,
VeraCrypt supports so-called \emph{hidden volumes}:
an encrypted volume whose content and metadata is indistinguishable
from the free space of a primary encrypted volume
hosting the hidden volume.
This protection is limited, however, because
a primary volume can host only a single hidden volume.
A potential coercer might therefore assume \emph{by default} 
that the coercee has a hidden volume,
and interpret a claim of non-possession of the decryption keys
as a refusal to provide them.
\purbs might enhance coercion protection by enabling an encrypted volume 
to contain any number of hidden volumes,
facilitating a stronger ``$N+1$'' defense.
Even if a coercee reveals up to $N$ ``decoy'' volumes,
the coercer cannot know whether there are any more.

\subsection{Integrated Encryption Scheme}
\label{sec:ies}

The Integrated Encryption Scheme (IES)~\cite{abdalla01oracle} is a hybrid 
encryption scheme that enables the encryption of arbitrary message strings
(unlike ElGamal, which requires the message to be a group element),
and offers flexibility in underlying primitives.
To send an encrypted message, a sender first generates an ephemeral 
Diffie-Hellman key pair and uses the public key of the recipient to derive a 
shared secret. The choice of the Diffie-Hellman group is flexible, \eg 
multiplicative groups of integers or elliptic curves.
The sender then relies on a cryptographic hash 
function to derive the shared keys used to encrypt the message with a 
symmetric-key cipher and to compute a MAC using the encrypt-then-MAC 
approach.
The resulting ciphertext is structured as shown in Figure~\ref{fig:ies}.

\begin{figure}[ht!]
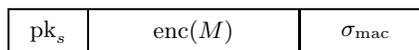

	\centering
	\begin{bytefield}[boxformatting={\centering\small}, bitwidth=1em]{16}
		\bitbox{3}{$\text{pk}_s$} & 
		\bitbox{8}{$\text{enc}(M)$} & 
		\bitbox{5}{$\sigma_{\text{mac}}$}
	\end{bytefield}
	\caption{Ciphertext output of the Integrated Encryption Scheme where 
	$\text{pk}_s$ is an ephemeral public key of the sender, and 
	$\sigma_{\text{mac}}$ 
	and $\text{enc}(M)$ are generated using the DH-derived keys.}
	\label{fig:ies}
	\vspace{-0.4cm}
\end{figure}

\section{Hiding Encryption Metadata}
\label{sec:encoding}

This section addresses the challenges of encoding and decoding
\emph{Padded Uniform Random Blobs} or \purbs
in a flexible, efficient, and cryptographically agile way.
We first cover notation, system and threat models, followed by
a sequence of strawman approaches that address different challenges on the 
path towards the full \mspurb scheme. We start with a scheme 
where ciphertexts are encrypted with a shared secret and addressed to 
a single recipient. We then improve it to support public-key operations with 
a single cipher suite, and finally to multiple recipients 
and multiple cipher suites.

\subsection{Preliminaries}


Let $\lambda$ be a standard security parameter.
We use \$ to 
indicate randomness, $\getrand$ to denote random sampling, 
$\parallel$ to denote string concatenation and 
|value| to denote the bit-length of ``value''. We write PPT as
an abbreviation for probabilistic polynomial-time.
Let $\Pi = (\E, \D)$ be 
an \ccatwo-secure authenticated-encryption (AE) 
scheme~\cite{bellare08authenticated} where $\E_{K}(m)$ and $\D_K(c)$ are 
encryption and decryption algorithms, respectively, given a 
message $m$, a ciphertext $c$, and a key $K$. 
Let $\mac = (\M, \V)$ be strongly unforgeable Message Authentication 
Code (MAC) generation and verification algorithms.
An authentication tag generated by \mac must be indistinguishable from a 
random bit string.

Let $\G$ be a cyclic finite group of prime order $p$ generated by the
group element $g$ where the gap-CDH problem
is hard to solve (\eg an elliptic 
curve or a multiplicative group of integers modulo a large prime).
Let $\hide:\ \G(1^\lambda) \to \{0, 1\}^\lambda$
be a mapping that encodes a group element of $\mathbb{G}$ to a binary 
string that is indistinguishable from a uniform random bit string
(\eg Elligator~\cite{bernstein13elligator}, Elligator 
Squared~\cite{tibouchi14squared, aranha14binary}). Let $\texttt{Unhide:}\ 
\{0, 1\}^\lambda \to \mathbb{G}(1^\lambda)$ be the counterpart to 
\texttt{Hide} that decodes a binary string into a group element of 
$\mathbb{G}$.

Let $\hasha:\G \to \{0, 1\}^{2\lambda}$ and $\hashb:\{0, 1\}^{*} 
\to \{0, 1\}^{2\lambda}$ be two distinct cryptographic hash functions.
Let $\texttt{PBKDF} : \{salt, password\} \to \{0, 1\}^{2\lambda}$ be a secure 
password-based key-derivation function~\cite{percival09stronger, 
biryukov15argon2, krawczyk2010cryptographic}, a ``slow'' hash function 
that converts a salt and a password into a bit string that can 
be used as a key for symmetric encryption.

\subsubsection{System Model}

Let \emph{data} be an application-level unit of data (\eg
a file or network message). 
A \emph{sender} wants to send an encrypted version of data to one 
or more \emph{recipients}.
We consider two main approaches for secure data exchanges: 

(1) Via \emph{pre-shared secrets}, where the sender shares with the 
recipients
long-term one-to-one passphrases $\hat{S}_1,..., \hat{S}_r$ that the
participants can use in a password-hashing scheme to derive ephemeral 
secrets $S_1,..., S_r$.

(2) Via \emph{public-key cryptography}, where sender and recipients derive
ephemeral secrets $Z_i=\hasha(X^{y_i}) = \hasha({Y_i}^x)$ using a hash 
function $\hasha$.  Here $(x, X=g^x)$ denotes the sender's one-time (private, 
public) key pair and $(y_i, Y_i = g^{y_i})$ is the key pair of recipient $i\in 1,...,r$.

In both scenarios, the sender uses ephemeral secrets 
$S_1,...,S_r$ or $Z_1,...,Z_r$ to encrypt (parts of) the \purb 
header using an authenticated encryption (AE) scheme.

We refer to a tuple $\suite = \langle \G, p, g, \hide(\cdot), \Pi, \hasha, 
\hashb \rangle$ used in the \purb 
generation as a \emph{cipher suite}. 
This can be considered similar to the notion of 
a cipher suite in TLS~\cite{rfc5246}.
Replacing any component of a suite (\eg the group) results in a different 
cipher suite.

\subsubsection{Threat Model and Security Goals}

We will consider two different types of computationally bounded adversaries:

\begin{compactenum}
	\item An \emph{outsider} adversary who does not hold a private key or a 
	password valid for decryption; 
	\item An \emph{insider} adversary who is a ``curious'' and active 
	legitimate recipient with a valid decryption key.
\end{compactenum}
Both adversaries are adaptive.

\vspace{0.1cm}
\noindent Naturally, the latter adversary has more power, \eg she can 
recover the plaintext payload. Hence, we consider different security 
goals given the adversary type:
\begin{compactenum}
	\item We seek \ccatwo security against the outsider adversary,
	\ie the encoded 
	content and \emph{all metadata} must be indistinguishable from random 
	bits under an adaptive chosen-ciphertext attack;
	\item We seek recipient privacy~\cite{barth06privacy}
	against the insider 
	adversary under a chosen-plaintext attack, \ie a recipient 
	must not be able to determine the identities of the ciphertext's 
	other recipients.
\end{compactenum}
Recipient privacy is a generalization of the key indistinguishability 
notion~\cite{bellare01key} where an adversary is unable to 
determine whether a given public key has been used for a given encryption.

\subsubsection{System Goals}

We wish to achieve two system goals beyond security:

\begin{compactitem}
     \item \purbs must provide cryptographic agility. They should 
     accommodate either one or multiple recipients, allow
     encryption for each recipient using a shared password or a public key, 
     and support different cipher suites.
     Adding new cipher suites must be seamless and must not 
     affect or break backward compatibility with other cipher suites.
	\item \purbs' encoding and decoding must be ``reasonably'' efficient.
	In particular, the number of expensive public-key operations should be 
	minimized, and padding must not impose excessive space overhead.
\end{compactitem}

\subsection{Encryption to a Single Passphrase}
\label{sec:singlepass}

We begin with
a simple strawman \purb encoding format
allowing a sender to encrypt $\text{data}$ 
using a single long-term passphrase $\hat{S}$ shared 
with a single recipient (\eg out-of-band via a secure channel).
The sender and recipient use an agreed-upon
cipher suite defining the scheme's components. 
The sender first generates a fresh symmetric session key $K$ and computes the 
PURB payload as $\E_K(\text{data})$.
The sender then generates a random salt and derives the ephemeral 
secret $S = \texttt{PBKDF}(\text{salt}, \hat{S})$. 
The sender next creates an~\emph{entry point} ($EP$) containing the session 
key~$K$, the position of the payload and potentially other metadata.
The sender then encrypts the EP using $S$.
Finally, the sender concatenates the three segments to form the \purb
as shown in Figure~\ref{fig:onepass}.

\begin{figure}[ht!]
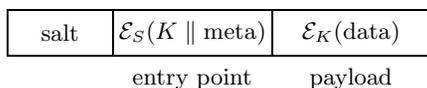

	\centering
	\begin{bytefield}[boxformatting={\centering\small}, bitwidth=1em]{16}
		\bitbox{4}{salt} & \bitbox{6}{$\E_{S}(K \parallel 
        \text{meta})$} & \bitbox{6}{$\E_{K}$($\text{data}$)} \\
		\bitbox[t]{4}{} & \bitbox[t]{6}{entry point} & 
		\bitbox[t]{6}{payload}
	\end{bytefield}
	\caption{A \purb addressed to a single recipient and encrypted 
	with a passphrase-derived ephemeral secret $S$.}
	\label{fig:onepass}
\end{figure}

\subsection{Single Public Key, Single Suite}
\label{sec:singlepk}

We often prefer to use public-key cryptography,
instead of pre-shared secrets,
to establish secure communication or encrypt data at rest. 
Typically the sender or initiator indicates in the file's cleartext metadata 
which public key this file is encrypted for (\eg in PGP), 
or else parties exchange public-key certificates in cleartext during
communication setup (\eg in TLS). 
Both approaches generally leak the receiver's identity.
We address this use case with a second strawman \purb encoding format 
that builds on the last by enabling the 
decryption of an entry point~$EP$ using a private key. 

To expand our scheme to the public-key scenario, we adopt the 
idea of a hybrid asymmetric-symmetric scheme from the IES (see 
\S\ref{sec:ies}).
Let $(y, Y)$ denote the recipient's key pair. The sender generates an 
ephemeral key pair $(x, X)$, computes the ephemeral secret 
$Z=\hasha(Y^x)$, then proceeds as before, 
except it encrypts $K$ 
and associated metadata with $Z$ instead of $S$.
The sender replaces the salt in the 
\purb with her encoded ephemeral public key $\hide(X)$, where 
$\hide(\cdot)$ maps a group element to a uniform random bit string. The 
resulting \purb structure is shown in Figure~\ref{fig:onepk}.

\begin{figure}[ht!]
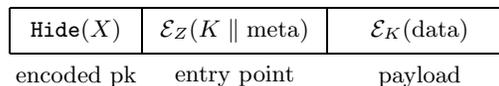

	\centering
	\begin{bytefield}[boxformatting={\centering\small}, bitwidth=1em]{19}
		\bitbox{5}{$\hide(X)$} & 
        \bitbox{7}{$\E_{Z}(K \parallel \text{meta})$} & 
		\bitbox{7}{$\E_{K}$($\text{data}$)} \\
		\bitbox[t]{5}{encoded pk} & \bitbox[t]{7}{entry point} & 
		\bitbox[t]{7}{payload}
	\end{bytefield}
	\caption{A \purb addressed to a single recipient that uses a  
		public key $Y$, where $X$ is the public key of the sender 
		and $Z=\hasha(Y^x)$ is the ephemeral 
		secret.}
	\label{fig:onepk}
\end{figure}

\subsection{Multiple Public Keys, Single Suite}
\label{sec:mulkeys}

We often wish to encrypt a message to several recipients,
\eg in multicast communication or mobile group chat.
We hence add support for encrypting one message under 
multiple public keys that are of the same suite.

As the first step, we adopt the idea of multi-recipient public-key 
encryption~\cite{kurosawa2002multi, bellare07multi} where the sender 
generates a single 
key pair and uses it to derive an ephemeral secret with each of the 
intended recipients. The sender creates one entry point per recipient. These 
entry points contain the same session key and metadata but are encrypted 
with different ephemeral secrets. 

As a~\purb's purpose is to prevent metadata leakage, 
including the number of recipients, a~\purb cannot reveal how many entry 
points exist in the header. Yet a legitimate recipient needs to have a way to 
enumerate possible candidates for her entry point. Hence, the primary 
challenge is to find a space-efficient layout of entry points---with no 
cleartext markers---such that the recipients are able to find their segments 
efficiently.

\para{Linear Table.} The most space-efficient approach is to place 
entry points sequentially. 
In fact, OpenPGP suggests a similar approach for achieving 
better privacy~\cite[Section~5.1]{rfc4880}. However, in 
this case, decryption is inefficient: the recipients have to attempt sequentially to decrypt each 
potential entry point, before finding their own or reaching the end of the 
\purb.

\para{Fixed Hash Tables.} A more computationally-efficient 
approach is to use a hash table of a fixed size. 
The sender creates a hash table and places each encrypted 
entry point there, identifying the corresponding position by hashing an 
ephemeral secret. 
Once all the entry points are placed, the remaining slots 
are filled with random bit strings, hence a third-party is unable to 
deduce the number of recipients.
The upper bound, corresponding to the size of the hash table,
is public information.
This approach, however, yields significant space overhead:
in the common case of a single 
recipient, all the unpopulated slots are filled with random bits but still 
transmitted.
This approach also has the downside of imposing
an artificial limit on the number of recipients.

\para{Expanding Hash Tables.} We therefore include not 
one but 
a sequence of hash tables whose sizes are consecutive
powers of two. Immediately following the encoded public 
key, the sender encodes a hash table of length one, followed (if needed) 
by a hash table of length two, one of length four, etc.,
until all the entry points 
are placed. Unpopulated slots are filled with random bits.
To decrypt a \purb, a recipient decodes the public key 
$X$, derives the ephemeral secret, computes the 
hash index in the first table (which is always zero), and tries to decrypt the 
corresponding entry point. On failure, the recipient moves to 
the second hash table, seeks the correct position and tries again, and so 
on.

\para{Definitions.}
We now formalize this scheme.
Let $r$ be the number of recipients and 
$({y_1}, {Y_1}), \ldots,({y_r}, {Y_r})$ be their corresponding key pairs.
The sender generates a fresh key pair $(x, X)$ and computes one
ephemeral secret $k_i= \hasha({Y_i}^x) $ per recipient.
The sender uses a second hash function $\hashb$ to derive independent 
encryption keys as $Z_i = \hashb(\textnormal{``key''} \parallel k_i)$ and 
position 
keys as $P_i = \hashb(\textnormal{``pos''} \parallel k_i)$.
Then the sender encrypts the data and creates $r$~entry 
points $\E_{Z_1}{(K, \text{meta})},..., \E_{Z_r}{(K, \text{meta})}$. The position 
of an entry in a hash table $j$ is $(P_i \mod 2^{j})$.
The sender iteratively tries 
to place an entry point in HT0 (hash table 0), then in HT1, and so on, until 
placement succeeds (\ie no collision occurs). 
If placement fails in the last existing hash table HT$j$, 
the sender appends another hash table HT$(j+1)$ of size $2^{j+1}$ and 
places the entry point there. 
An example of a \purb encrypted for five recipients is illustrated in 
Figure~\ref{fig:mulpk}.

\begin{figure}[ht!]
	\centering
	\begin{bytefield}[boxformatting={\centering\small}, bitwidth=.95em]{25}
		
		\bitbox[]{5}{encoded pk} & 
		\bitbox[]{5}{HT0} 
		& \bitbox[]{5}{HT1} & 
		\bitbox[]{5}{HT2} &
		\bitbox[]{5}{payload} \\
		
		\bitbox{5}{$\hide(X)$} & 
		\bitbox{5}{$\E_{Z_1}(K)$} 
		& \bitbox{5}{$\E_{Z_3}(K)$} & 
		\bitbox{5}{$\E_{Z_4}(K)$} &
		\bitbox{5}{$\E_{K}$($\text{data}$)} \\
		
		\bitbox[rt]{10}{} & \bitbox{5}{$\E_{Z_2}(K)$} & 
		\bitbox{5}{\color{gray} random} & \bitbox[l]{5}{} \\
		
		\bitbox[r]{15}{} & \bitbox[lr]{5}{$\E_{Z_5}(K)$} & 
		\bitbox[l]{5}{} \\
		
		\bitbox[r]{15}{} & \bitbox[lrtb]{5}{\color{gray} random} & \bitbox[l]{5}{} 
	\end{bytefield}
	\caption{A \purb with hash tables of increasing sizes (HT0, HT1, HT2). 
	Five and two slots of the hash tables are filled with entry points and
		random bit strings respectively.
		The metadata ``meta'' in the entry points is 
		omitted from the figure. Hash-table entries are put one after another in 
		the byte representation of a \purb.}
	\label{fig:mulpk}
\end{figure}

To decode, a recipient reads the public key; derives the ephemeral 
secret $k_i$, the encryption key $Z_i$ and the position 
key  $P_i$; and iteratively tries matching positions in hash tables until 
the decryption of the entry point succeeds. 
Although the recipient does not initially know
the number of hash tables in a \purb, the 
recipient needs to do only a single expensive public-key operation, and the 
rest are inexpensive symmetric-key decryption trials. 
In the worst case of a small message encrypted to many recipients,
or a non-recipient searching for a nonexistent entry point,
the total number of trial decryptions required
is logarithmic in the \purb's size.

In the common case of a single recipient, only a single hash table of 
size $1$ exists, and the header is compact. With $r$ recipients, the 
worst-case compactness is having $r$ hash tables (if each insertion leads 
to a collision), which happens with exponentially decreasing probability.
The expected number of trial decryptions is $\log_2 r$.

\subsection{Multiple Public Keys and Suites}
\label{sec:mulsuites}

In the real world, not all data recipients' keys
might use the same cipher suite. 
For example, users might prefer different key lengths or might use 
public-key algorithms in different groups.
Further, we must be able to introduce new cipher suites gradually,
often requiring larger and differently-structured keys and ciphertexts,
while preserving interoperability and compatibility with old cipher suites.
We therefore build on the above strawman schemes to produce
{\em Multi-Suite PURB} or \mspurb,
which offers cryptographic agility by supporting the encryption of data for 
multiple different cipher suites.

When a \purb is multi-suite encrypted, the recipients need a way to learn 
whether a given suite has been used and where the encoded 
public key of this suite is located in the \purb.
There are two obvious approaches to enabling recipients
to locate encoded public keys for multiple cipher suites:
to pack the public keys linearly at the beginning of a \purb,
or to define a fixed byte position for each cipher suite.
Both approaches incur undesirable overhead. In the former case, the recipients 
have to check all possible byte ranges, performing an expensive 
public-key operation for each. The latter approach results in 
significant space overhead and lack of agility, as unused fixed positions 
must be filled with random bits, and adding new cipher suites 
requires either assigning progressively larger fixed positions
or compatibility-breaking position changes to existing suites.

\para{Set of Standard Positions.} To address this challenge, we introduce a 
\emph{set} of standard byte positions per suite. 
These sets are public and 
standardized for all \purbs. The set refers to positions where the suite's 
public key could be in the \purb. For 
instance, let us consider a suite 
\texttt{PURB\_X25519\_AES128GCM\_SHA256}.
We can define---arbitrarily 
for now---the set of 
positions as $\{0, 64, 128, 1024\}$.
As the length of the encoded public key is fully defined by the suite 
($32$~bytes here, as Curve25519 is used), the recipients will iteratively 
try to decode a public key at $[0{:}32)$, then $[64{:}96)$, etc. 

If the sender wants to encode a \purb for two suites A and B, she 
needs to find one position in each set such that the public keys do not 
overlap. For instance, if $\text{set}_A = \{0, 128, 256\}$ and $\text{set}_B = 
\{0, 32, 64, 128\}$, and the public keys' lengths are $64$ and $32$, 
respectively, one possible choice would be to put the public key for suite 
A in $[0{:}64)$, and the public key for suite $B$ in $[64{:}96)$.
All suites typically have position $0$ in their set, so that in the common 
case of a \purb encoded for only one suite, the encoded public key 
is at the beginning of the \purb for maximum space efficiency.
Figure~\ref{fig:multirec-multisuites} illustrates an example encoding.
With well-designed sets,
in which each new cipher suite is assigned at least one position
not overlapping with those assigned to prior suites,
the sender can encode a \purb for any subset of the suites.
We address efficiency hereunder, and
provide a concrete example with real suites in 
Appendix~\ref{appendix:allowed-positions}.

\begin{figure}[ht!]
	\centering
	\begin{bytefield}[boxformatting={\centering\small}, bitwidth=1em]{24}
		
		\bitbox[]{6}{encoded $\text{pk}_A$} & 
		\bitbox[]{3}{HT0} 
		& \bitbox[]{5}{HT1} & 
		\bitbox[]{5}{HT2} &
		\bitbox[]{5}{payload} \\
		
		\bitbox{7}{$\hide(X_A)$} & 
		\bitbox{2}{{\color{gray} rnd}} 
		& \bitbox{5}{$\hide(X_B)$} & 
		\bitbox{5}{$\E_{Z_2}(K)$} &
		\bitbox{5}{$\E_{K}$($\text{data}$)} \\
		
		\bitbox[rt]{9}{} & \bitbox{5}{$\E_{Z_1}(K)$} & 
		\bitbox{5}{\color{gray} random} & \bitbox[l]{5}{} \\
		
		\bitbox[r]{14}{} & \bitbox[lr]{5}{$\E_{Z_3}(K)$} & 
		\bitbox[l]{5}{} \\
		
		\bitbox[r]{14}{} & \bitbox[lrtb]{5}{\color{gray} random} & \bitbox[l]{5}{} 
	\end{bytefield}
	\caption{Example of a \purb encoded for three public keys in two suites 
		(suite $A$ and $B$). The sender generates one ephemeral key pair 
		per 
		suite ($X_A$ and $X_B$). In this example, $X_A$ is placed at the first 
		allowed position, 
		and $X_B$ moves to the second allowed position 
		(since the first position is taken by suite A). Those positions are public 
		and fixed for each suite. HT0 cannot be used for storing an 
		entry point, as $X_A$ partially occupies it; HT0 is 
		considered ``full'' and the entry point is placed in subsequent hash 
		tables - here HT1.}
	\label{fig:multirec-multisuites}
\end{figure}

\para{Overlapping Layers.}
One challenge is that suites might indicate different lengths for both 
their public keys and entry points.
An encoder can easily accommodate this requirement by 
processing each suite used in a \purb as an independent logical layer.
Conceptually, each layer is composed of the public key and the entry-point 
hash tables for the recipients that use a given suite,
and all suites' layers overlap.
To place the~layers, an encoder first initializes a byte layout for the \purb.
Then, she reserves in the byte layout the~positions for the public keys of 
each suite used. Finally, she fills the hash tables of each suite with 
corresponding entry points. She identifies whether a given hash-table slot 
can be filled by checking the byte layout; the bytes might already be 
occupied by an entry point of the same or a different suite or one of the 
public keys. 
The hash tables for each suite start immediately after 
the suite public key's first possible position. Thus, upon reception of a 
\purb, a decoder knows exactly where to start decryption trials.
The payload is placed right after the 
last encoded public key or hash table, and its start position is recorded in 
the meta in each entry point.

\para{Decoding Efficiency.} 
We have not yet achieved our decoding efficiency goal, however:
the recipient must perform several expensive public-key operations
for each cipher suite,
one for each potential position until the correct position is found.
We reduce this overhead to a single public-key operation per suite by 
removing the recipient's need to know in which of the suite positions 
the public key was actually placed.
To accomplish this, a sender XORs bytes at 
all the suite positions and places the result into one of them.
The sender first constructs the whole PURB as before, then she substitutes 
the bytes of the already-written encoded public key with the 
XOR of bytes at all the defined suite positions (if they do not exceed the 
\purb length), which could even correspond to encrypted payload.
To decode a \purb, a recipient starts by reading and XORing the values at 
{\em all} the positions defined for a suite.
This results in an encoded public key, if 
that suite was used in this \purb. 

\para{Encryption Flexibility.}
Although multiple cipher suites can be used in a \purb,
so far these suites 
must agree on one payload encryption scheme, as a payload appears 
only once. 
To lift this constraint, we decouple encryption schemes for entry points and 
payloads.
An~entry-point encryption scheme is a part of a cipher suite,
whereas a payload encryption scheme is indicated separately
in the metadata ``meta'' in each entry point.

\subsection{Non-malleability}
\label{sec:mallea}

Our encoding scheme \mspurb so far ensures integrity only of the payload 
and the entry point a decoder uses.
If the entry points of other recipients or random-byte fillings are 
malformed, a decoder will not detect this.
If an attacker obtains access to a decoding oracle, he can 
randomly flip bits in an intercepted \purb, query the oracle on decoding 
validity, and learn the structure of the \purb
including the exact length of the payload.
An example of exploiting malleability is the Efail 
attacks~\cite{poddebniak18efail}, which tamper with PGP- or 
S/MIME-encrypted e-mails to achieve exfiltration of the plaintext.

To protect \purbs from undetected modification,
we add integrity protection to \mspurb using a MAC 
algorithm. A sender derives independent encryption 
$\enckey = \hashb(\text{``enc''} \parallel K)$ and MAC $\mackey = 
\hashb(\text{``mac''} \parallel K)$ keys from the encapsulated key $K$, and 
uses \mackey to compute an authentication tag over a full \purb as the final 
encoding step.
The sender records the utilized MAC algorithm in the meta in the entry 
points, along with the payload encryption scheme that now does not need 
to be authenticated.
The sender places the tag at the very end of the \purb,
which covers the entire \purb including encoded public keys,
entry point hash tables, payload ciphertext, and any padding required.

Because the final authentication tag covers the entire \purb,
the sender must calculate it after all other \purb content is finalized,
including the XOR-encoding of all the suites' public key positions.
Filling in the tag would present a problem, however,
if the tag's position happened to overlap with one of the 
public key positions of some cipher suite,
because filling in the tag would corrupt
the suite's XOR-encoded public key.
To handle this situation,
the sender is responsible for ensuring that the authentication 
tag does not fall into any of the possible public key positions
for the cipher suites in use.

To encode a \purb, a sender prepares entry points, lays out the 
header, encrypts the payload, adds padding (see \S\ref{sec:pad}), and 
computes the~\purb's total length.
If any of the byte positions of the authentication tag to be appended
overlap with public key positions,
the sender increases the padding to next bracket, until the public-key 
positions and the tag are disjoint.
The sender proceeds with XOR-encoding all suites' public keys,
and computing and appending the tag.
Upon receipt of a \purb, a decoder computes the potential public keys, 
finds and decrypts her entry point, learns the decryption scheme 
and the MAC algorithm with the size of its tag.
She then verifies the \purb's integrity and decrypts the payload.

\subsection{Complete Algorithms}
\label{sec:purb-complete}

We summarize the encoding scheme by giving detailed 
algorithms.
We begin by defining helper \headpurb algorithms that encode and decode
a \purb header's data for a single cipher suite.
We then use these algorithms in defining the final \mspurb encoding 
scheme.

Recall the notion of a cipher suite $\suite = \langle \G, p, g, \hide(\cdot), 
\Pi, \hasha, \hashb \rangle$,
where $\G$ is a cyclic group of order $p$ generated by $g$;
\hide is a mapping: $\G \to \{0, 1\}^\lambda$;
$\Pi = (\E, \D)$ is an authenticated-encryption scheme;
and $\hasha:\G \to \{0, 1\}^{2\lambda}$, $\hashb:\{0, 1\}^{*} 
\to \{0, 1\}^{2\lambda}$ are two distinct cryptographic hash functions.
Let $sk$ and $pk$ be a private key and a public key, respectively, for 
$\langle \G, p, g \rangle$ defined in a cipher suite.
We then define the full \headpurb and \mspurb algorithms as follows:

\begin{algo}[Hdr\purb]
	\label{algo:hdrpurb}
	\hangtwo
	$\hdr\encap(R, \suite) \to (\tau, k_1, \ldots, k_r)$: Given a set 
	of public keys $R=\{pk_1=Y_1,\ldots, pk_r = Y_r\}$ of a suite~$\suite$:
	\begin{compactenum}[\hspace{1em} (1)]
		\item Pick a fresh $x \in \mathbb{Z}_{p}$ and compute $X = 
      g^{x}$ where $p, g$ are defined in \suite.
		\item Derive $k_1 =  \hasha(Y_1^{x}),\ldots, k_r = \hasha(Y_r^{x})$.
		\item Map $X$ to a uniform string $\tau_X = \hide(X)$.
		\item Output an encoded public key $\tau = \tau_X$ and 
		$k_1, \ldots, k_r$.
	\end{compactenum}
	
	\hangone
	$\hdr\decap(sk(\suite), \tau) \to k$:  Given a private 
	key $sk = y$ of a suite $\suite$ and an encoded public key $\tau$:
	\begin{compactenum}[\hspace{1em} (1)]
		\item Retrieve $X = \unhide(\tau)$.
		\item Compute and output $k = \hasha(X^{y})$.
	\end{compactenum}
\end{algo}

\begin{algo}[\mspurb]
	\label{algo:purbs}
	\hangtwo
	$\psetup(\onelambda) \to \suite$: Initialize a cipher suite 
	$\suite = \langle \G, p, g, \hide(\cdot), \Pi, \hasha, \hashb \rangle$.
	
	\hangone
	$\pkeygen(\suite) \to (sk, pk)$: Given a suite $\suite = \langle \G, p, g, 
	\ldots \rangle$, pick $x \in \mathbb{Z}_{p}$ and 
	compute $X = g^{x}$. Output $(sk = x, pk = X)$.
	
	\hangone
	$\penc(R, m) \to c$: 
	Given a set of public keys of an indicated suite 
	$R = \{pk_1(\suite_1), \ldots, pk_r(\suite_r)\}$ and a message $m$:
	\begin{compactenum}[\hspace{1em} (1)]
		\item Pick an appropriate symmetric-key encryption 
      scheme $(\enc, \dec)$ with key length $\lambda_K$, a MAC 
		algorithm $\mac = (\M, \V)$, and a hash function $\hashpayload:\{0, 
		1\}^{*} \to \{0, 1\}^{\lambda_K}$ such that the key length $\lambda_K$
    matches the security level of the most conservative suite.
		\item Group $R$ into $R_1, \ldots, R_n$, \st all public keys in a 
		group $R_i$ share the same suite $\suite_i$. Let $r_i = |R_i|.$
		\item For each $R_i$:
		\begin{compactenum}[(a)]
			\item Run $(\tau_i, k_1, \ldots, k_{r_i}) = \hdr\encap(R_i, \suite_i)$;
			\item Compute entry-point keys
			$\keys_i = (Z_1 = \hashb(\textnormal{``key''} \parallel k_1),\ldots, 
			Z_{r_i} = \hashb(\textnormal{``key''} \parallel k_{r_i}))$ 
			and positions $\aux_i = (P_1 = \hashb(\textnormal{``pos''} \parallel 
			k_1),\ldots, P_{r_i} = \hashb(\textnormal{``pos''} \parallel k_{r_i}))$.
		\end{compactenum}
		\item Pick $K \getrand \{0, 1\}^{\lambda_K}$.
		\item Record $(\enc, \dec)$, \mac and \hashpayload in \meta.
		\item Compute a payload key $\enckey = 
		\hashpayload(\text{``enc''} \parallel K)$ and a MAC key $\mackey = 
		\hashpayload(\text{``mac''} \parallel K)$.
		\item Obtain $\ctxtpayload = \enc_{\enckey}(m)$.
		\item Run $c' \gets \hyperref[algo:layout]{\layout}(\tau_1, \ldots, \tau_n, 
		\keys_1,\allowbreak \ldots,\allowbreak \keys_n,\allowbreak \aux_1, 
		\allowbreak \ldots, \aux_n,
		\suite_1, \ldots, \suite_n, K, \meta, \ctxtpayload)$
		(see Algorithm~\ref{algo:layout} on page~\pageref{algo:layout}).
		\item Derive an authentication tag $\sigma = \M_{\mackey}(c')$ 
		and output $c = c' \parallel \sigma$.
	\end{compactenum}
	
	\hangone
	$\pdec(sk(\suite), c) \to m/\bot$: Given a private key $sk$ of a suite 
	\suite and a ciphertext $c$:
	\begin{compactenum}[\hspace{1em} (1)]
		\item Look up the possible positions of a public key 
		defined by $\suite$ and XOR bytes at all the positions to obtain the 
		encoded public key $\tau$.
		\item Run $k \gets \hdr\decap(sk, \tau)$.
		\item Derive $Z = \hashb(\textnormal{``key''} \parallel k)$ and $P = 
		\hashb(\textnormal{``pos''} \parallel k)$.
		\item Parse $c$ as growing hash tables and, using the secret 
		$Z$ as the key, trial-decrypt the entries defined by $P$ to obtain $K 
		\parallel \meta$.
		If no decryption is successful, return $\bot$.
		\item Look up the hash function \hashpayload, a $\mac = (\M, \V)$ 
		algorithm and the length of \mac output 
		tag $\sigma$ from \meta. Parse $c$ as $\langle c' \parallel \sigma 
		\rangle$. Derive 
		$\mackey = \hashpayload(\text{``mac''} \parallel K)$ and run 
		$\V_{\mackey}(c', \sigma)$. On failure, return $\bot$.
		\item Derive $\enckey = \hashpayload(\text{``enc''} \parallel K)$, read 
		the start and the end of the payload from $\meta$ (it is written by 
		\layout) 
		to parse $c'$ as $\langle hdr \parallel \ctxtpayload \parallel 
		\text{padding} \rangle$, and return $\dec_{\enckey}(\ctxtpayload)$ 
		where $\dec$ is the payload decryption algorithm specified in 
		$\meta$.
	\end{compactenum}
\end{algo}


\begin{theorem}
	If for each cipher suite
	$\suite = \langle \G, p, g,\allowbreak \hide(\cdot),\allowbreak \Pi,\allowbreak \hasha, 
	\hashb \rangle$
	used in a PURB we have that:
	the gap-CDH problem is hard relative to $\G$, \textnormal{\hide} maps 
	group elements in $\G$ to uniform random strings, 
	$\Pi$ is \textnormal{\ccatwo}-secure,
	and $\hasha$, $\hashb$ and $\hashpayload$ are modeled as a random 
	oracle;
	and moreover that \mac is strongly unforgeable with its MACs being 
	indistinguishable from random, and the~scheme for 
	payload encryption $(\enc, \dec)$ is \textnormal{\cpa}-secure,
	then \mspurb is \textnormal{\ccatwo}-secure against an~outsider 
	adversary.
	\label{theorem:outsider}
\end{theorem}

\begin{proof}
	See Appendix~\ref{proof:outsider}.
\end{proof}

Theorem~\ref{theorem:outsider} also implies that an outsider adversary 
cannot break recipient privacy under an \ccatwo attack, as long as the two 
possible sets of recipients $N_0, N_1$ induce the same distribution on the 
length of a \purb.

\begin{theorem}
	If for each cipher suite
	$\suite = \langle \G, p, g,\allowbreak \hide(\cdot),\allowbreak 
	\Pi,\allowbreak \hasha, \hashb \rangle$,
	used in a PURB we have that:
	the gap-CDH problem is hard relative to $\G$, \textnormal{\hide} maps 
	group elements in $\G$ to uniform random strings, 
	$\Pi$ is \textnormal{\ccatwo}-secure,
	$\hasha$ and $\hashb$ are modeled as a random oracle, and
	the order in which cipher suites are used for encoding is fixed; 
	then \mspurb is recipient-private against an~\textnormal{\cpa} insider 
	adversary.
	\label{theorem:insider}
\end{theorem}

\begin{proof}
	See Appendix~\ref{proof:insider}.
\end{proof}

\subsection{Practical Considerations}
\label{sec:practical}

Cryptographic agility (\ie changing the encryption scheme) for the 
payload is provided by the metadata embedded in the entry points. For 
entry points themselves, we recall that the recipient uses trial-decryption 
and iteratively tests suites from a known, public, ordered list.
To add a 
new suite, it suffices to add it to this list. With this technique, a \purb does 
not need version numbers. There is, however, a trade-off between the 
number of supported suites and the maximum decryption time.
It is important that a sender follows the fixed order of the 
cipher suites during encoding because a varying order might result in a 
different header length, given the same set of recipients and 
sender's ephemeral keys, which could be used by an insider 
adversary.

If a nonce-based authenticated-encryption scheme is used for entry points, 
a sender needs to include a distinct \emph{random} 
nonce as a part of entry-point ciphertext (the nonce of each entry point 
must be unique per \purb).
Some schemes, \eg AES-GCM~\cite{bellare16multiAE}, have been shown 
to retain their security when the same nonce is reused with different keys.
When such a scheme is used, there can be a single \emph{global} nonce 
to reuse by each entry point. However, generalizing this approach of a 
global nonce to any scheme requires further analysis.

\para{Hardening Recipient Privacy.}
The given instantiation of \mspurb provides recipient privacy only under a 
\emph{chosen-plaintext} attack.
If information about decryption success is leaked, an insider adversary 
could learn identities of other recipients of a \purb by 
altering the header, recomputing the MAC, and querying candidates.
A possible approach to achieving \ccatwo recipient privacy is  
to sign a complete \purb using a strongly existentially unforgeable signature 
scheme and to store the verification key in each entry point, as similarly 
done in the broadcast-encryption scheme by Barth et 
al.~\cite{barth06privacy}.
This approach, however,
requires adaptation to the multi-suite settings, and it will 
result in a significant increase of the header size and decrease in efficiency.
We leave this question for future work.

\para{Limitations.}
The \mspurb scheme above is not secure against quantum 
computers, as it relies on discrete logarithm hardness.
It is theoretically possible to substitute IES-based key 
encapsulation with a quantum-resistant variant to achieve quantum \ccatwo 
security. The requirements for
substitution are \ccatwo security and compactness (it must be possible to 
securely reuse sender's public key to derive shared secrets with multiple 
recipients).
Furthermore, as \mspurb is non-interactive, they do not offer forward 
secrecy.

Simply by looking at the sizes (of the header for a malicious insider, or the 
total size for a malicious outsider), an adversary can infer a bound on the 
total number of recipients. 
We partially address this with padding in 
\S\ref{sec:pad}. However, no reasonable padding scheme can 
perfectly hide this information. If this is a problem in practice, we suggest 
adding dummy recipients.

Protecting concrete implementations against timing 
attacks is a highly challenging task.
The two following properties are required for basic hardening.
First, the implementations of \purbs should always attempt to  decrypt all 
potential entry points using all the recipient's suites.
Second, decryption errors of any source as well as inability to recover the 
payload should be processed in constant time and always return $\bot$.

\section{Limiting Leakage via Length}
\label{sec:pad}

The encoding scheme presented above in \S\ref{sec:encoding} produces 
blobs of data that are indistinguishable from random bit-strings of the same 
length, thus leaking no information to the adversary directly via their content.
The length itself, however,
might indirectly reveal information about the content.
Such leakage is
already used extensively in traffic-analysis attacks, \eg website 
fingerprinting~\cite{panchenko11website,dyer12peek,wang13improved,wang16realistically},
video identification~\cite{reed2016leaky, schuster2017beauty, 
reed2017identifying}, and VoIP traffic fingerprinting~\cite{wright2007language, 
chang2008inferring}. Although solutions involving application- or network-level 
padding are numerous, they are typically designed for a specific 
problem domain, and the more basic problem of length-leaking ciphertexts 
remains. In any practical solution, some leakage is unavoidable.
We show, however,
that typical approaches such as padding to the size of a block cipher
are fundamentally 
insufficient for efficiently hiding the plaintext length effectively,
especially for plaintexts that may vary in size by orders of magnitude.

We introduce \padname, a novel padding scheme designed for,
though not restricted to, encoding \purbs.
\padme reduces length leakage for a wide range of 
encrypted data types,
ensuring asymptotically lower leakage of $O(\log\log M)$,
rather than $O(\log M)$ 
for common stream- and block-cipher-encrypted data. \padname's 
space overhead is moderate,
always less than $12\%$ and decreasing with file size.
The intuition behind \padname is to pad objects to lengths representable
as limited-precision floating-point numbers.
A \padme length is constrained in particular
to have no more significant bits (\ie information)
in its mantissa than in its exponent.
This constraint  limits information leakage
to at most double that of conservatively padding to the next power of two,
while reducing overhead through logarithmically-increasing precision
for larger objects.


\com{\color{blue}
Finally, our contribution is both a class of padding schemes with overhead 
decreasing with file size, and one specific 
instance, \padname, which has $O(\log \log L)$ leakage and around 
$12\%$ 
overhead. Should this prove inadequate (\eg insufficient for a specific dataset 
of sensitive objects), it is trivial to select other instances of the class 
with different leakage/overhead trade-off, while maintaining the 
overhead properties of the class.}

Many defenses already exist for specific scenarios, \eg against
website fingerprinting~\cite{dyer12peek, wang17walkie}. \padname 
does not attempt to compete with tailored solutions in their domains.
Instead, \padme aims for a substantial increase
in application-independent length leakage protection
as a generic measure of security/privacy hygiene.

\subsection{Design Criterion}

We design \padname again using intermediate strawman approaches for clarity.
To compare these straightforward alternatives with our proposal,
we define a game where an adversary guesses the plaintext behind 
a padded encrypted blob. This game is inspired by related work
such as defending against a \emph{perfect attacker}~\cite{wang17walkie}.

\para{Padding Game.} 
Let $P$ denote a collection of plaintext objects of maximum length $M$:
\eg~data, documents, or application data units.
An honest user chooses a plaintext $p \in P$, then pads and encodes it into a \purb $c$.
The adversary knows almost everything: all possible plaintexts $P$, the \purb $c$ and the parameters used to generate it,
such as schemes and number of recipients.
The adversary lacks only the private inputs and decryption keys for $c$.
The adversary's goal is to guess the plaintext $p$ based on the observed \purb $c$ of length $|c|$.

\para{Design Goals.}
Our goal in designing the padding function is to manage both
space overhead from padding
and maximum information leaked to the adversary.

\subsection{Definitions}

\para{Overhead.}
Let $c$ be a padded ciphertext resulting from \purb-encoding plaintext $p$.
For simplicity we focus here purely on overhead incurred by padding,
by assuming an unrealistic, ``perfectly-efficient''
\purb encoding that (unlike \mspurb) incurs no space overhead
for encryption metadata.
We define the {\em additive overhead} of $|c|$ over $|p|$
to be $|c|-|p|$, the number of extra bytes added by padding.
The {\em multiplicative overhead} of padding is $\frac{|c|-|p|}{|p|}$,
the relative fraction by which $|c|$ expands $|p|$.

\para{Leakage.}
Let $P$ be a finite space of plaintexts of maximum length $M$.
Let $f:\N \to \N$ be a padding function that yields the padded size $|c|$
given a plaintext length $|p|$, for $p\in P$.
The image of $f$ is a set $R$ of padded lengths that $f$ can produce
from plaintexts $p \in P$.
\\
We quantify the leakage of padding function $f$
in terms of the number of elements in $R$.
More precisely, we define the leakage as the number of bits (amount of information entropy)
required to distinguish a unique element of $R$, which is $\ceil{\log_2 |R|}$. Intuitively, a function that pads everything to a constant size larger than all plaintexts (\eg $f(p) = 1$~Tb)
leaks no information to the adversary, because $|R|=1$ (and observing $|c|=1$~Tb leaks no information about the plaintext), 
whereas more fine-grained padding functions leak more bits.
\com{	not quite sure what this is saying and it sounds like a claim
	that probably doesn't need to be made, so let's not make it. -baf
This definition makes no assumptions about the distribution of $P$,
as we desire a generic padding function.
}

\subsection{Strawman Padding Approaches}

We first explore two strawman designs,
based on different padding functions $f$.
A padding function that offers any useful protection cannot be one-to-one,
otherwise the adversary could trivially invert it and recover $|p|$.
We also exclude randomized padding schemes for simplicity,
and because in practice adversaries can typically
cancel out and defeat random padding factors statistically
over many observations.
Therefore,
only padding functions that group many plaintext lengths
into fewer padded ciphertexts are of interest in our analysis.

\para{Strawman 1: Fixed-Size Blocks.}
We first consider a padding function $ f(L) = b \cdot \ceil{L/b} $,
where $b$ is a block size in bytes.
This is how objects often get ``padded'' by default in practice,
\eg in block ciphers or Tor cells.
In this case, the 
\purb's size is a multiple of $b$,
the maximum additive overhead incurred is $b-1$ bytes,
and the leakage is $\ceil{\log_2~M/b} = O(\log M)$,
where $M$ is the maximum plaintext size.

In practice, 
when plaintext sizes differ by orders of magnitude, there is no 
good value for $b$ that serves all plaintexts well. For 
instance, consider 
$b=1$\,MB. Padding small files and network messages would 
incur a large overhead: \eg padding Tor's $512$\,B cells
to $1$\,MB would incur overheads of $2000\times$. 
In contrast, padding a $700$\,MB movie with at most $1$\,MB of chaff 
would add only a little confusion to the adversary,
as this movie may still be readily distinguishable from others by length.
To reduce information leakage asymptotically
over a vast range of cleartext sizes,
therefore, padding must depend on plaintext size.

\para{Strawman 2: Padding to Powers of 2.}
The next step is to pad to varying-size blocks,
which is the basis for our actual scheme.
The intuition is that for small plaintexts, 
the blocks are small too, yielding modest overhead, whereas for larger 
files, blocks are larger and group more plaintext lengths together,
improving leakage asymptotically.
A simple approach is to pad plaintexts into buckets $b_i$
of size varying as a power of some base, \eg two,
so $b_i = 2^i$.
The padding function is thus $f(L) = 2^{\ceil{\log_2 L}}$.
We call this strawman \npot.

Because \npot pads plaintexts of maximum length $M$
into at most $\ceil{\log_2 M}$ buckets,
the image $R$ of $f$ contains only $O(\log M)$ elements.
This represents only $O(\log \log M)$ bits of entropy or information leakage,
a major asymptotic improvement over fixed-size blocks.

The maximum overhead is substantial, however,
almost +$100\%$: \eg a 17~GB Blu-Ray movie would be padded into 32~GB.

Using powers of another base $x>2$,
we reduce leakage further at a cost of more overhead:
\eg padding to the nearest 
power of $3$ incurs overhead up to +$200\%$,
with less leakage but still $O(\log\log M)$. 
We could reduce overhead by using a fractional base $1 < x < 2$,
but fractional exponents are cumbersome in practical padding functions
we would prefer to be simple and operate only on integers.
Although this second strawman succeeds in achieving
asymptotically lower leakage than padding to fixed-size blocks,
it is less attractive in practice
due to high overhead when $x \ge 2$
and due to computation complexity when $1 < x < 2$.

\subsection{\padname}
\label{subsec:padme}

We now describe our padding scheme \padname, which limits information 
leakage about the length of the plaintext for wide range of encrypted data 
sizes. Similarly to 
the previous strawman, \padname also asymptotically leaks $O(\log\log M)$ 
bits of information, but its overhead is much lower (at most $12\%$ and 
decreasing with $L$).


\para{Intuition.} In \npot, 
any permissible padded length $L$ has the form $L = 2^n$.
We can therefore represent $L$ as a binary floating-point number with
a $\floor*{\log_2 n}+1$-bit exponent and a mantissa of zero,
\ie no fractional bits.

In \padname, we similarly represent a permissible padded length
as a binary floating-point number, but we allow 
a non-zero mantissa at most as long as the exponent (see Figure 
\ref{table-padme-vs-power2}). This approach doubles the number of 
bits used to represent an allowed padded length --
hence doubling absolute leakage via length --
but allows for more fine-grained buckets, reducing overhead. 
\padname asymptotically leaks the same number of bits as \npot,
differing only by a constant factor of $2$,
but reduces space overhead
by almost $10\times$ (from +$100\%$ to +$12\%$). 
\com{	I still don't understand or believe this claim. -baf
More importantly, the buckets sizes now grow logarithmically with respect 
to $L$, instead of growing exponentially as in \npot. 
Thus, the overhead in percentage decreases with $L$.
}
More importantly, the multiplicative expansion overhead decreases with $L$
(see Figure \ref{fig:plot-padme-vs-pow2-percentage}).

\begin{figure}[tb]
	\centering
	\begin{bytefield}[boxformatting={\centering\small}, bitwidth=1em]{24}
		\bitbox{12}{$\floor*{\log_2 n}+1$-bit exponent} & \bitbox[l]{12}{0-bit mantissa} \\
	\end{bytefield}
	\vspace{-1cm}
	\caption*{In the strawman \npot, the allowed length $L = 2^n$ can 
		be represented as a binary floating-point number with a
		$\floor*{\log(n)+1}$ bits of exponent and
		no mantissa.}
	\vspace{0.5cm}
	
	\begin{bytefield}[boxformatting={\centering\small}, bitwidth=1em]{24}
		\bitbox{12}{$\floor*{\log_2 n}+1$-bit exponent} & \bitbox{12}{$\floor*{\log_2 n}+1$-bit 
			mantissa} \\
	\end{bytefield}
	\vspace{-1cm}
	\caption{\padname represents lengths as floating-point numbers,
		allowing the mantissa to be of at most 
		$\floor*{\log_2 n}+1$ bits.}
	\label{table-padme-vs-power2}
	\vspace{1cm}
\end{figure}

\com{
\para{Intuition.} Unlike the previous strawman design, where the bucket sizes 
$b_i$ were growing exponentially, in \padname the bucket sizes grow more slowly 
-- they are logarithmically-sized with respect to the plaintext size -- keeping 
the overhead low. This design yields the same asymptotic leakage but with a 
much lower overhead in practice. Figure \ref{fig:plot-padme-vs-pow2-percentage} 
shows the different bucket design between \padname and padding to the next 
power of two. \padname achieves this bucket design by fixing a number of 
low-order bits to $0$ on the padded length, effectively reducing the number of 
allowed lengths, and grouping ciphertexts in buckets. 
}

\begin{figure}[t]
	\vspace{-0.5cm}
	\centering
	\includegraphics[width=0.75\linewidth]{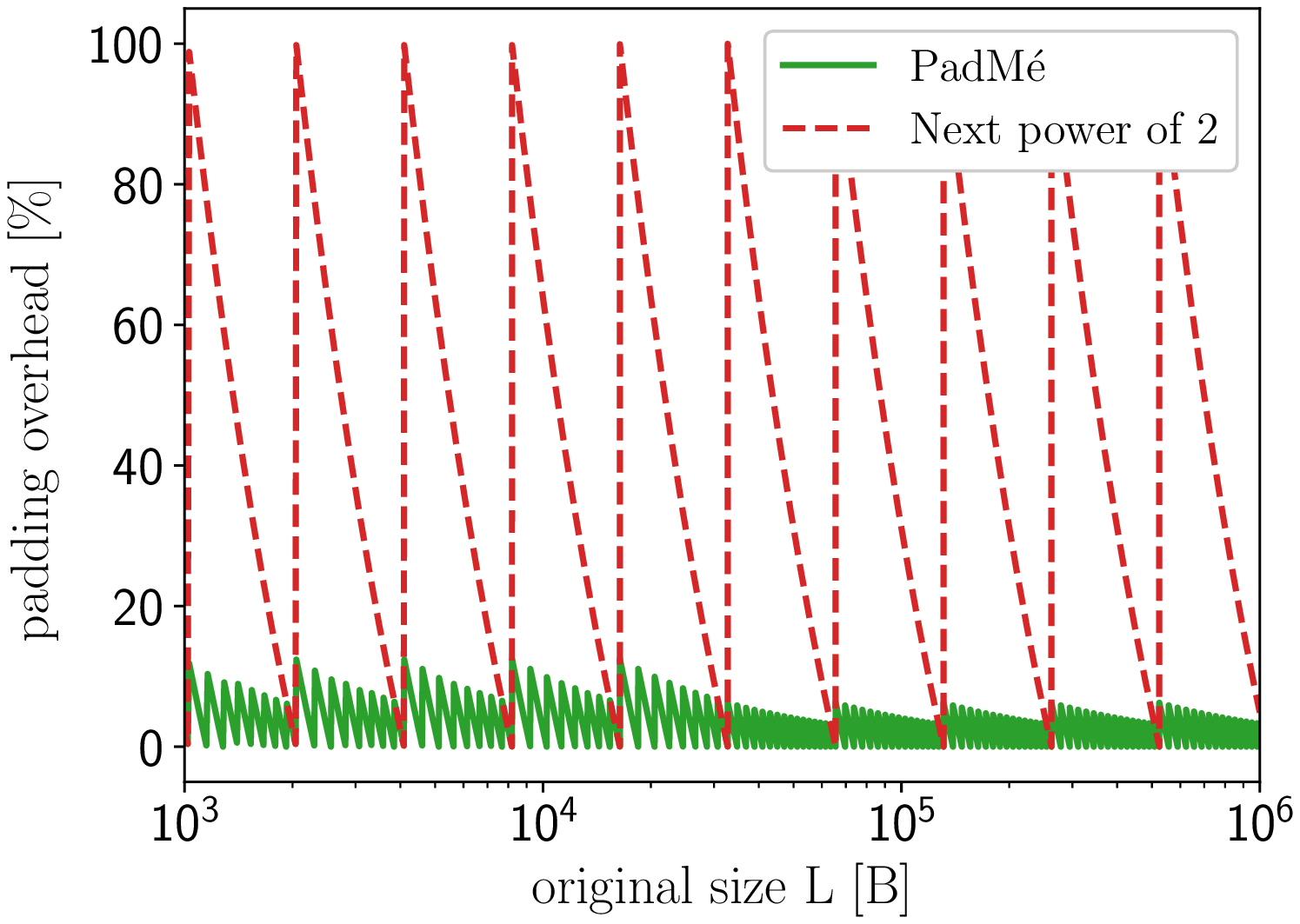}
	\vspace{-0.2cm}
	\caption{Maximum multiplicative expansion overhead
		with respect to the plaintext size $L$. The 
		naïve approach to pad to the next power of two has a constant 
		maximum overhead of $100\%$,
		whereas \padname's maximum overhead
		decreases with $L$, following $\frac{1}{2\log_2 L}$.}
	\label{fig:plot-padme-vs-pow2-percentage}
\end{figure}

\para{Algorithm.} To compute the padded size $L' = f(L)$,
ensuring that its 
floating-point representation fits in at most $2\times\floor*{\log_2 n}+1$ bits, 
we require the last $E-S$ bits of $L'$ to be $0$.
$E = \floor*{\log_2 L}$ is the 
value of the exponent, and $S = \floor*{\log_2 E}+1$ is the size of 
the exponent's binary 
representation. The reason for the substraction will become clear later. For 
now, we demonstrate how $E$ and $S$ are computed in 
Table~\ref{table-ieee}.
\begin{table}[tb]
	\centering
	
	\caption{The IEEE floating-point representations of $8$, $9$ and $10$. 
		The value $8$ has $1$ bit of mantissa (the initial 1 is omitted), and $2$ 
		bits of exponents; $9$ has a $3$-bits mantissa and a $2$-bit exponent, 
		while the 
		value $10$ as $2$ bits of mantissa and exponents. \padme enforces the 
		mantissa to be no longer than the exponent, hence $9$ gets rounded up 
		to the next permitted length $10$.}
	\label{table-ieee}
	
	\begin{tabular}{lllllrll}
		L & L & E & S & & \multicolumn{3}{l}{IEEE representation} \\ 
		\hline
		8 & {\normalfont 0b}1000 & 3 & 2 & & {\normalfont 0b}1.0 &* 
		2\textasciicircum {\normalfont 0b}11 & \\ 
		9 & {\normalfont 0b}1001 & 3 & 2 & & {\normalfont 0b}1.001 &* 
		2\textasciicircum {\normalfont 0b}11 & \\ 
		10 & {\normalfont 0b}1010 & 3 & 2 & & {\normalfont 0b}1.01 &* 
		2\textasciicircum {\normalfont 0b}11 & \\ 
	\end{tabular}
\end{table} 

Recall that \padname requires the mantissa's bit length to be no 
longer than that of the exponent. In Table \ref{table-ieee}, for the 
value $L=9$ the mantissa is longer than the exponent: it is 
``too precise'' and therefore not a permitted padded length. 
The value $10$ is permitted, however, so a $9$ byte-long ciphertext is
padded to $10$ bytes.

To understand why \padname requires the low $E-S$ bits to be $0$, notice that
forcing all the last $E$ bits to $0$ is equivalent to padding to
a power of two. In comparison, \padname allows $S$ extra bits to represent 
the padded size, with $S$ defined as the bit length of the exponent.

Algorithm~\ref{algo:padme} specifies the \padme function precisely.
\com{	This appears to be no longer applicable,
	and  redundant since LAYOUT now generates random padding. -baf
Once the padded size $L'$ is computed, a 
\purb plaintext of 
length $L$ is simply 
padded with $L'-L~0$'s (to be 
precise, we suggest following the compact padding scheme ISO/IEC 
7816-4:2005\footnote{\href{https://www.iso.org/standard/36134.html} 
{https://www.iso.org/standard/36134.html}}).
}

\newcommand\mycommfont[1]{\footnotesize\ttfamily\textcolor{blue}{#1}}
\SetCommentSty{mycommfont}

\begin{algorithm}
	\SetKwRepeat{Do}{do}{while}
	\DontPrintSemicolon
	\SetAlgoLined
	\KwData{length of content $L$}
	\KwResult{length of padded content $L'$}
	$E \gets \floor*{\log_2 L}$ \tcp*{$L$'s floating-point exponent}
	$S \gets \floor*{\log_2 E}+1$ \tcp*{\# of bits to represent $E$}
	$z \gets E - S$ \tcp*{\# of low bits to set to 0}
	$ m \gets (1 \ll z) - 1$ \tcp*{mask of $z$ 1's in LSB}
	 \tcp*{round up using mask $m$ to clear last $z$ bits}
	$L' \gets (L + m)~ \&~ {\raise.17ex\hbox{$\scriptstyle\sim$}}m$
	\caption{\padname}
	\label{algo:padme}
\end{algorithm}

\para{Leakage and Overhead.}
By design,
if the maximum plaintext size is $M$,
\padme's leakage is $O(\log\log M)$ bits,
the length of the binary representation of the largest plaintext.
As we fix $E-S$ bits to $0$ and round up, the 
maximum overhead is $2^{E-S}-1$.
We can estimate the maximum multiplicative overhead as follows:
\begin{equation}\label{eq3}
\begin{aligned}[b]
\text{max overhead} & = \frac{2^{E-S}-1}{L}
 < \frac{2^{E-S}}{L} \\
& \approx \frac{2^{\floor*{\log_2 L}-\floor*{\log_2\log_2 L}-1}}{L} \\
& \approx \frac{1}{2\cdot2^{\log_2\log_2 L}} \\
& = \frac{1}{2\log_2 L}
\end{aligned}
\end{equation}


Thus, \padname's maximum multiplicative overhead decreases with respect to the 
file size $L$. The maximum overhead is +$11.\overline{11}\%$, when padding a $9$-byte file into $10$ bytes. For bigger files, the overhead is smaller.

\para{On Optimality.} There is no clear sweet spot on the 
leakage-to-overhead 
curve. We could easily force the last $\frac{1}{2}(E-S)$ bits to be 
$0$ 
instead of the last $E-S$ bits, for example, to 
reduce overhead and increase leakage. Still, what matters in 
practice is the relationship between $L$ and the overhead.
We show in \S\ref{sec:padme-eval} how this choice performs
with various real-world datasets.

\color{black}

\section{Evaluation}
\label{sec:eval}

Our evaluation is two-fold. First, we show the performance and overhead of 
the \purb encoding and decoding.  Second, using several datasets, we 
show how \padname facilitates hiding information about data length.

\subsection{Implementation}

We implemented a prototype of the~\purb
encoding and padding schemes in Go. The implementation follows 
the algorithms in \S\ref{sec:purb-complete}, and it consists 
of $2$\,kLOC. 
Our implementation relies on the open-source Kyber 
library\footnote{\href{https://github.com/dedis/kyber}{https://github.com/dedis/kyber}}
for cryptographic operations.
The code is designed to be easy to integrate with existing applications.
The code is still proof-of-concept, however, and has 
not yet gone through rigorous analysis and hardening,
in particular against timing 
attacks.

\para{Reproducibility.} All the datasets, the source code for \purbs and 
\padname, as well as scripts for reproducing all experiments, are available 
in the main repository\footnote{\label{footnote:purbsurl}\purbsurl}.

\subsection{Performance of the \purb Encoding}

The main question we answer in the evaluation of the encoding 
scheme is whether it has a reasonable cost, in terms 
of both time and space overhead, and whether it scales gracefully with an 
increasing number of recipients and/or cipher suites.
First, we measure the average CPU time required to encode and decode a \purb. Then, we compare the decoding 
performance with the performance of plain and anonymized OpenPGP 
schemes described below. Finally, we show how the compactness of the 
header changes with multiple recipients and suites, as a percentage of 
useful bits in the header.

\para{Anonymized PGP.} In standard PGP, the identity---more precisely, 
the public key ID---of the recipient is embedded in the header of the 
encrypted blob. This plaintext marker speeds up decryption, but enables a 
third party to enumerate all data recipients. In the so-called 
anonymized or ``hidden'' version of PGP~\cite[Section~5.1]{rfc4880}, this 
key ID is 
substituted with zeros. In this 
case, the recipient sequentially tries the encrypted entries of the header 
with her keys. We use the hidden PGP variant as a comparison 
for \purbs, which also does not indicate key IDs in the header but uses a 
more efficient structure. The hidden PGP variant still leaks the 
cipher suites used, the total length, and other plaintext markers (version 
number, etc.).

\subsubsection{Methodology}

We ran the encoding experiments on a consumer-grade laptop, 
with a quad-core 2.2 GHz Intel Core i7 processor and 16\,GB of RAM, 
using Go 1.12.5. 
To compare with an OpenPGP implementation, we use and modify 
Keybase's 
fork\footnote{\href{https://github.com/keybase/go-crypto}{https://github.com/keybase/go-crypto}}
 of the default Golang crypto 
library\footnote{\href{https://github.com/golang/crypto}{https://github.com/golang/crypto}},
 as the fork adds support for the ECDH scheme on Curve25519.

We further modify Keybase's implementation to add the support for the 
anonymized OpenPGP scheme. All the encoding 
experiments use a \purb suite based on the Curve25519 elliptic-curve group, 
{AES128\nobreakdash-GCM} for 
entry point encryption and SHA256 for hashing. We also apply the 
global nonce optimization, as discussed in \S\ref{sec:practical}.
For experiments needing more than one suite,
we use copies the above suite to ensure 
homogeneity across timing experiments.
The payload size in each experiment is $1$\,KB.
For each data point, we generate a new set of keys, one per recipient.
We measure each data point 20 times, using fresh randomness each time, 
and depict the median value and the standard deviation.

\subsubsection{Results}
\label{sec:evalperform}

\begin{figure*}[h]
	\centering
	\begin{subfigure}[t]{.49\linewidth}
		\includegraphics[width=\textwidth]{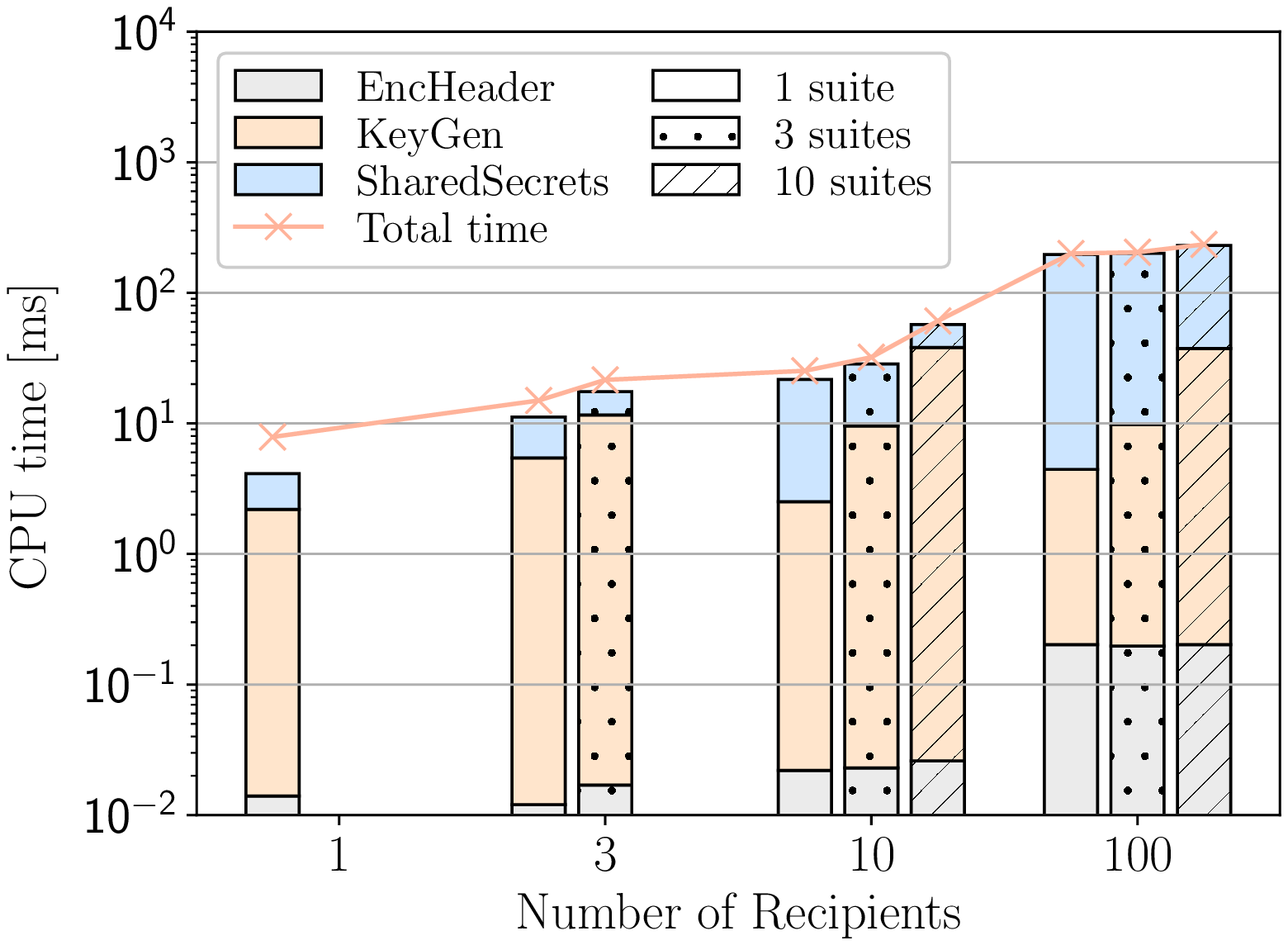}
		\caption{The CPU cost of encoding a \purb given the 
		number of recipients and of cipher suites. EncHeader: 
		encryption of entry points; KeyGen: generation and 
		hiding of public keys; SharedSecrets: computation of shared secrets.}
		\label{fig:enctime}
	\end{subfigure}\hfill
	\begin{subfigure}[t]{.49\linewidth}
		\includegraphics[width=\textwidth]{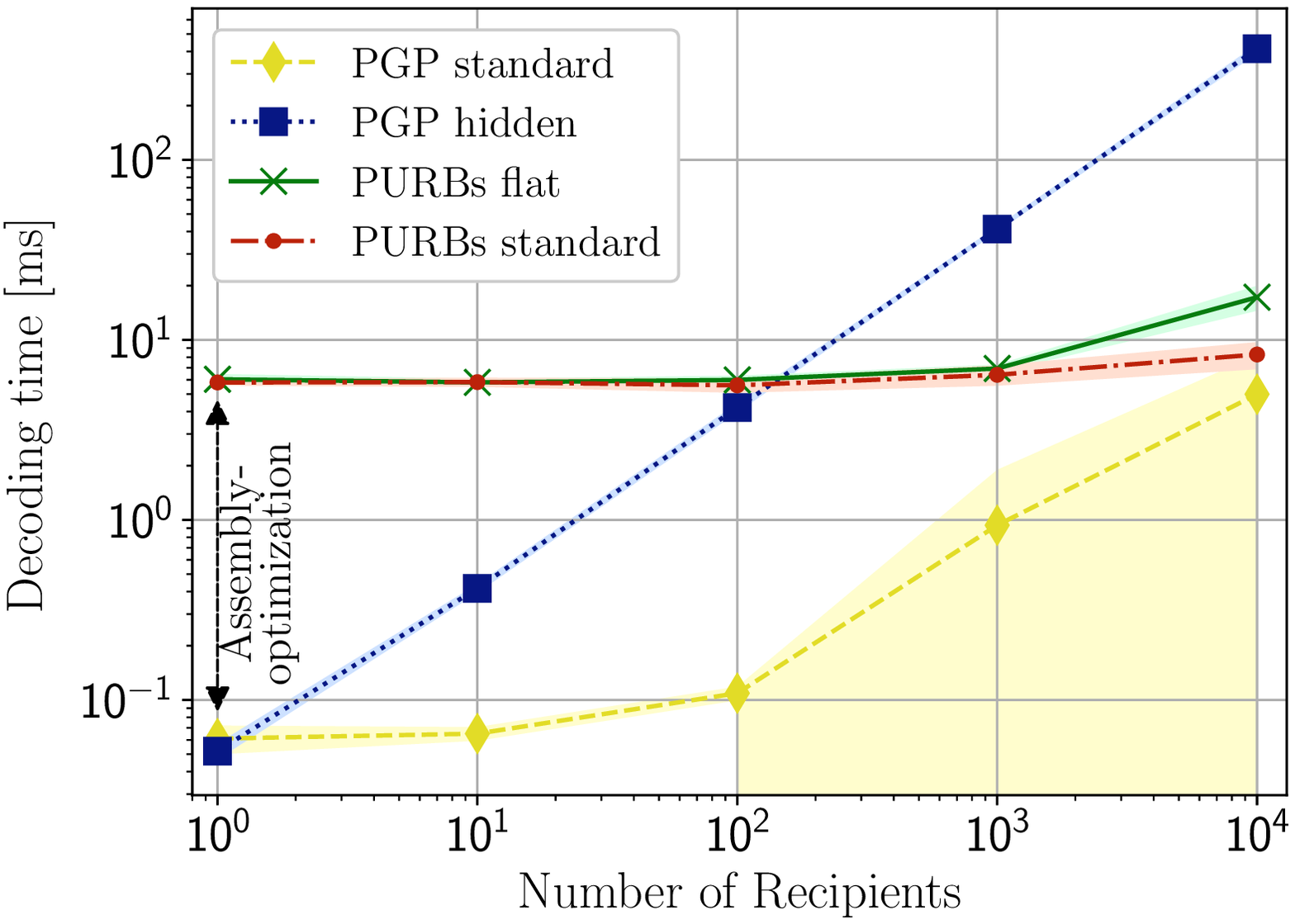}
		\caption{The worst-case CPU cost of decoding for PGP, PGP 
			with hidden recipients, PURBs without hash tables (flat), and 
			standard PURBs.}
		\label{fig:decode}
	\end{subfigure}%
	\captionsetup{justification=centering}
	\caption{Performance of the \purbs encoding.}
\end{figure*}

\begin{figure}[h]
	\includegraphics[width=\linewidth]{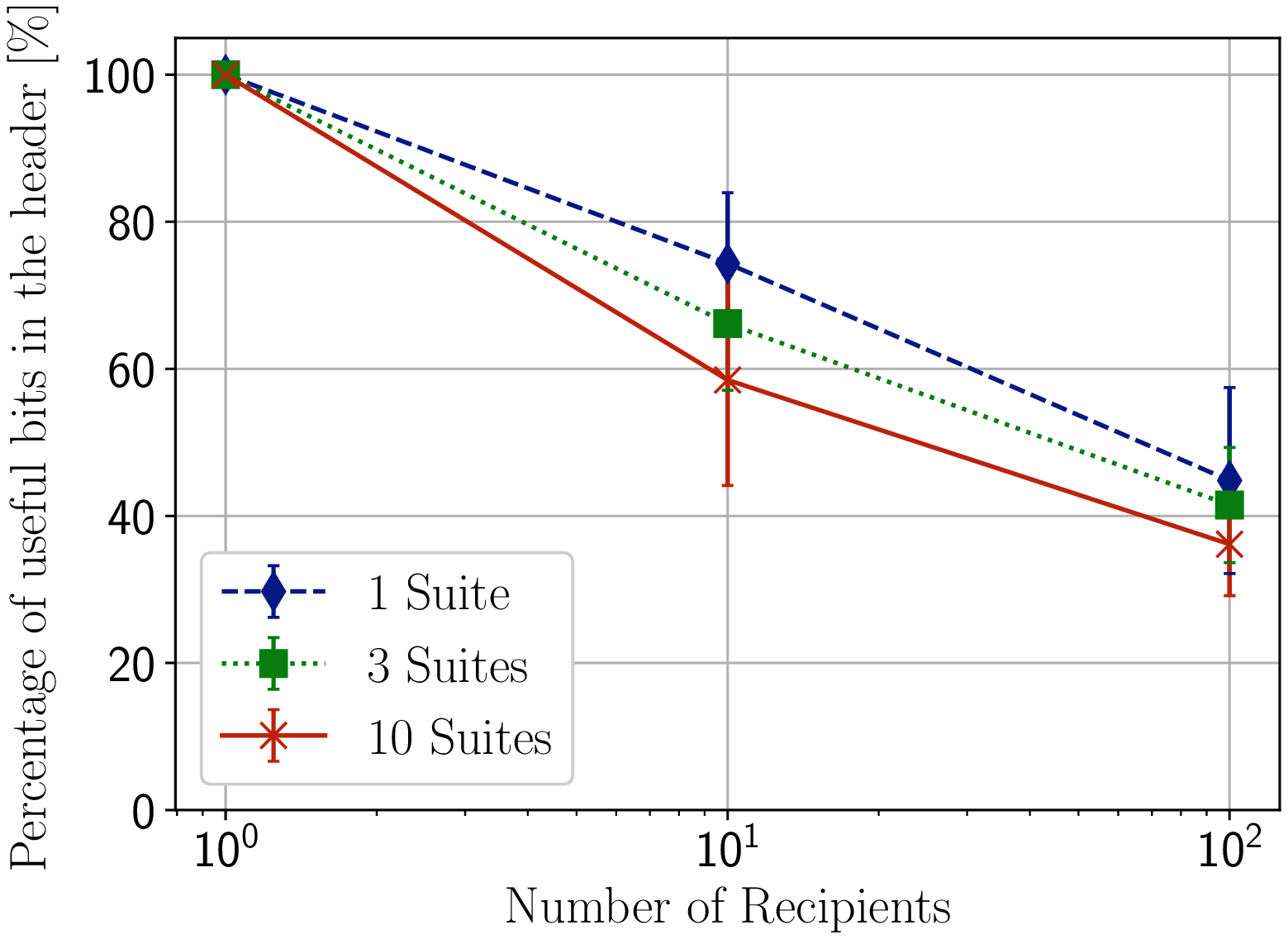}
	\caption{Compactness of the PURB header (\% of 
	non-random bits).}
	\label{fig:compactness}
\end{figure}

\para{Encoding Performance.} In this experiment, we first evaluate how the 
time required to encode a \purb changes with a growing number of 
recipients and cipher suites, and second, how the main computational 
components contribute to this duration. 
We divide the total encoding time into three components. The first is 
authenticated encryption of entry points.
The second is the generation and Elligator encoding of sender's public 
keys, one per suite. 
A public key is derived by multiplying a base point with a freshly generated 
private key (scalar). If the resultant public key is not encodable, which 
happens in half of the cases, a new key is generated. 
Point multiplication dominates this component, constituting 
$\approx 90\%$ of the total time. 
The third is the derivation of a shared secret with each recipient, 
essentially a single point-multiplication per recipient.
Other significant components of the total encoding duration are payload 
encryption, MAC computation and layout composition.
We consider cases using one, three or ten cipher suites. When more 
than one cipher suite is used, the recipients are equally divided among
them.

Figure~\ref{fig:enctime} shows that in the case of a single recipient, the 
generation of a public key and the computation of a shared secret 
dominate the total time and both take $\approx 2$\,ms. 
As expected, computing shared secrets starts dominating the total time 
when the number of recipients grows, whereas the duration of the 
public-key generation only depends on a number of cipher suites used. 
The encoding is arguably efficient for most cases of communication, as 
even with hundred recipients and ten suites, the time for creating a \purb 
is $235$\,ms. 

\para{Decoding Performance.} We measure the worst-case CPU time 
required to decipher a standard PGP message, a PGP message with 
hidden recipients, a \emph{flat} \purb that has a flat layout of entry points 
without hash tables, and a standard \purb. We use the Curve25519 suite
in all the PGP and \purb schemes.

Figure~\ref{fig:decode} shows the results. The~OpenPGP library uses 
the assembly-optimized Go elliptic library for point multiplication, hence 
the multiplication takes $\approx 0.05$--$0.1$\,ms there, while it 
takes $\approx 2$--$3$\,ms in Kyber. 
This results in a~significant difference in 
absolute values for small numbers of recipients. But our primary interest 
is the dynamics of total duration. The time increase for 
anonymous PGP is linear because, in the worst case, a decoder has to 
derive as many shared secrets as there are recipients.
\purbs in contrast exhibit almost constant time,
requiring only a single multiplication 
regardless of the number of recipients. A~decoder still has to 
perform multiple entry-point trial decryptions, but one such operation would 
account for only $\approx 0.3\%$ of the total time in the 
single-recipient, single-suite scenario. The advantage of 
using hash tables, and hence logarithmically less symmetric-key 
operations, is illustrated by the difference between \purbs standard
and \purbs \textit{flat}, which is noticeable after $100$ recipients and will 
become more pronounced if point multiplication is optimized.

\para{Header Compactness.} 
Compared with placing the header elements linearly,
our expanding hash table design is less 
compact, but enables more efficient decoding.
Figure~\ref{fig:decode} shows
an example of this trade-off,
PGP hidden versus \purbs standard.

In Figure~\ref{fig:compactness}, we show the compactness, or the 
percentage of the \purb header that is filled with actual data, with respect 
to the number of recipients and cipher suites. Not surprisingly, an 
increasing number of recipients and/or suites increases the collisions and 
reduces compactness: $45\%$ for $100$ recipients and $1$ suite, $36\%$ 
for $100$ recipients and $10$ suites. 
In the most common case of 
having one recipient in one suite, however,
the header is perfectly compact. Finally, 
there is a trade-off between compactness and efficient decryption.
We can easily increase compactness
by resolving entry point hash table collisions linearly,
instead of directly moving to the next hash table. The 
downside is that the recipient has more entry points to try.

\subsection{Performance of \padname Padding}
\label{sec:padme-eval}

In evaluating a padding scheme, one important metric is  
overhead incurred in terms of bits added to the plaintexts. By design, 
\padname's overhead is bounded by $\frac{1}{2\cdot\log_2 L}$\com{  (see 
Figure~\ref{fig:plot-logL-over-L}) }. As discussed in 
\S\ref{subsec:padme}, 
\padname does not escape the typical overhead-to-leakage trade-off, 
hence \padname's novelty does not lie in this tradeoff.
Rather, the novelty lies in the practical relation 
between $L$ and the overhead. \padname's overhead is moderate,
at most +$12\%$ and much less for large \purbs.

A more interesting question is how effectively, given an arbitrary collection 
of plaintexts $P$, \padname hides which plaintext is 
padded. \padname was designed to work with an arbritrary collection of 
plaintexts $P$. It remains to be seen how \padname performs when 
applied to a specific set of plaintexts $P$, \ie with a distribution coming 
from the real world, and to establish how well it groups files into sets of 
identical length. In the next section, we experiment with four datasets 
made of various objects: a collection of Ubuntu packages, a set of 
YouTube videos, a set of user files, and a set of Alexa Top $1$M websites.

\subsubsection{Datasets and Methodology}

\begin{table}
\centering
\caption{Datasets used in the evaluation of anonymity provided by \padname.}
\vspace{-0.3cm}
\begin{tabular}{ll}
	Dataset & \# of objects \\ 
	\hline 
	\normalfont Ubuntu packages & \normalfont 56,517 \\ 
	\normalfont YouTube videos & \normalfont 191,250 \\ 
	\normalfont File collections & \normalfont 3,027,460 \\ 
	\normalfont Alexa top 1M Websites & \normalfont 2,627  \\ 
\end{tabular} 
\label{table:datasets}
\end{table}

The Ubuntu dataset contains $56{,}517$ unique packages, parsed from 
the 
official repository of a live Ubuntu $16.04$ instance. As packages can 
be referenced in multiple repositories, we filtered the list by name and 
architecture. The reason for padding Ubuntu software updates is that the 
knowledge of updates enables a local eavesdropper to build a list of 
packages and their versions that are installed on a machine. If some of 
the~packages are outdated and have known vulnerabilities, an adversary 
might 
use it as an attack vector. A percentage of software updates 
still occurs over un-encrypted connections, which is still an issue; 
but encrypted connections to software-update repositories also expose 
which distribution and the kind of update being done (security / 
restricted\footnote{Contains proprietary software and drivers.} / 
multiverse\footnote{Contains software restricted by copyright.} / etc). We 
hope that this unnecessary leakage will disappear in the near future.

The YouTube dataset contains $191{,}250$ unique videos, obtained by 
iteratively querying the YouTube API. One semantic video is generally 
represented by $2-5$ .webm files, which corresponds to various video 
qualities. Hence, each object in the dataset is a unique (video, quality) pair. 
We use this dataset as if the videos were downloaded in bulk rather than 
streamed; that is, we pad the video as a single file.  The argument for 
padding YouTube videos as whole files is that, as shown by related work  
\cite{reed2016leaky, schuster2017beauty, reed2017identifying}, 
variable-bitrate encoding combined with streaming leak which video is 
being watched. If YouTube wanted to protect the privacy of its users, it 
could re-encode everything to constant-bitrate encoding and still stream it, 
but then the total length of the stream would still leak information. 
Alternatively, it could adopt a model similar to that of the iTunes store, where 
videos have variable bit-rate but are bulk-downloaded; but
again, the total downloaded length would leak information, requiring some 
padding. Hence, we explore how unique the YouTube videos are by length
with and without padding.

The files dataset was constituted by collecting the~file sizes in 
the home directories (`\verb|~user/|')
of $10$ co-workers and contains $3{,}027{,}460$ of 
both personal files 
and configuration files. These files were collected 
on machines running Fedora, Arch, and Mac~OS~X. The argument 
for analyzing the uniqueness of those files is not to encrypt each file 
individually -- there is no point in hiding the metadata of a file if the 
file's location exposes everything about it, e.g. `\verb|~user/.ssh|' --
but rather to quantify the privacy gain when padding those objects.

Finally, the Alexa dataset is made of $2{,}627$\com{sizes corresponding 
to whole websites taken in order} websites from the Alexa Top 1M list. The 
size of each website is the sum of all the resources loaded by the 
webpage, which has been recorded by piloting a `chrome-headless' 
instance with a script, mimicking real browsing. 
One reason for padding whole websites -- as opposed to padding 
individual resources -- is that related work in website fingerprinting showed 
the importance of the total downloaded size~\cite{dyer12peek}. The 
effectiveness of \padname when padding individual resources, or for 
instance bursts~\cite{wang17walkie}, is left as interesting future 
work.\com{; we focus 
here on showing the uniqueness of the websites with respect to their total 
size.}

\subsubsection{Evaluation of \padname}

The distribution of the objects sizes for all the datasets is shown in 
Figure~\ref{fig:fig3-datasets-cdf-size}. Intuitively, it is harder for an 
efficient padding scheme to build groups of same-sized files when there 
are large objects in the dataset. Therefore, we expect the last $5\%$ to 
$10\%$ of the four datasets to remain somewhat unique, even after padding.

\begin{figure}[t]
	\vspace{-0.5cm}
	\centering
	\includegraphics[width=0.9\linewidth]{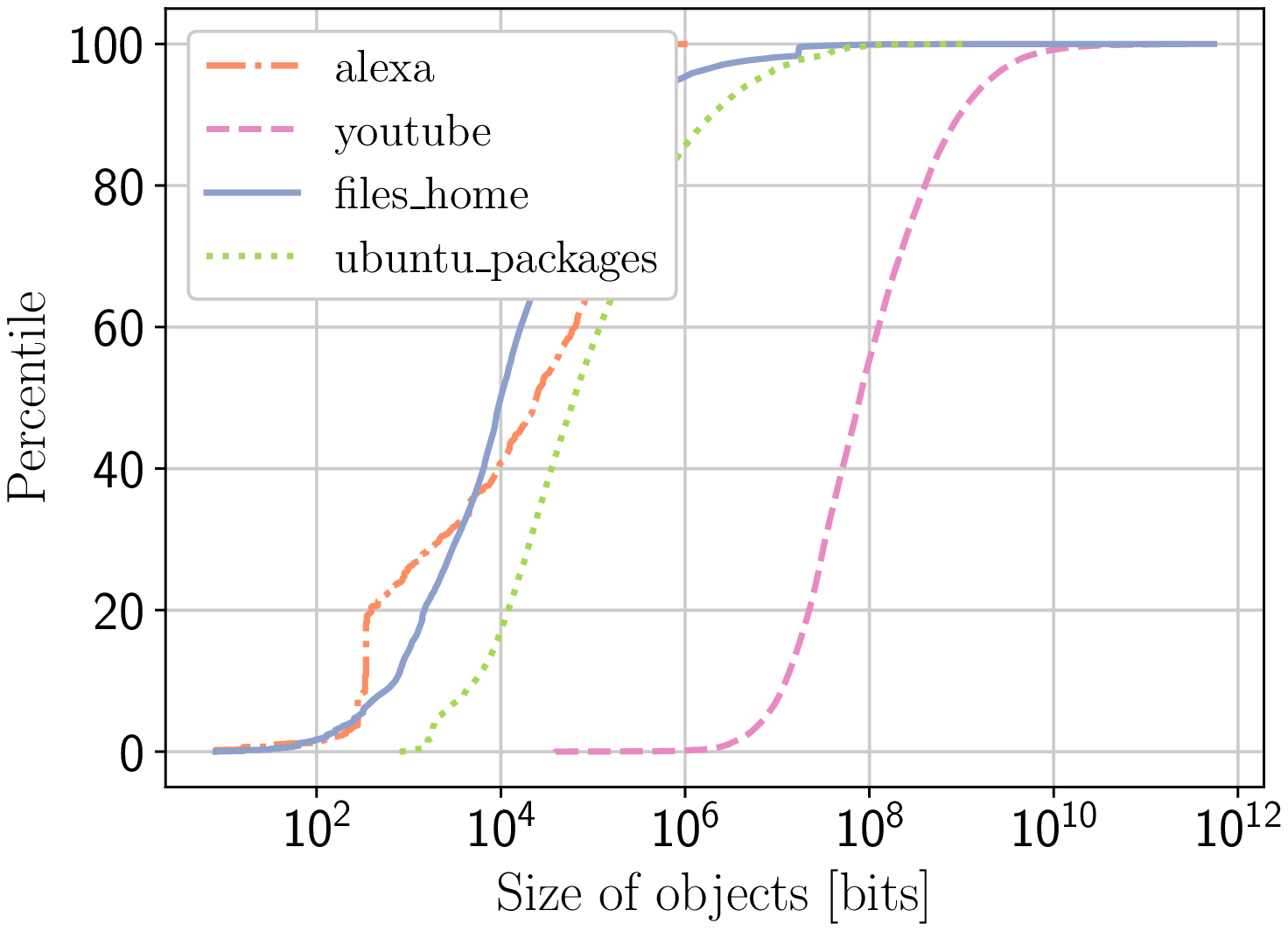}
	\vspace{-0.5cm}\caption{Distribution of the sizes of the objects in each dataset.}
	\label{fig:fig3-datasets-cdf-size}
\end{figure}

For each dataset, we analyze the 
anonymity set size of each object. To compute this metric, we group 
objects by their size, and report the distribution of the sizes of these 
groups. A large number of small groups indicate that many objects 
are easily identifiable. For each dataset, we 
compare three different approaches: the \npot strawman, \padname,
and padding to a fixed block size of 512B, like a Tor cell.
The anonymity metrics are 
shown in Figure~\ref{fig:padme-datasets}, and the respective overheads are 
shown in Table~\ref{table:padme-datasets-overhead}.

For all these 
datasets, despite containing very different objects, a large percentage of objects have a unique size: $87\%$ in the case of YouTube video (Figure~\ref{fig:padme-youtube}), $45\%$ in the case of files (Figure~\ref{fig:padme-files}), $83\%$ in the case of Ubuntu packages (Figure~\ref{fig:padme-ubuntu}),
and $	68\%$ in the case of Websites Figure~\ref{fig:padme-alexa}). 
These characteristics persist in traditional block-cipher encryption
(blue dashed curves) where objects are padded only to a block size. Even 
after being padded to $512$ bytes, the size of a Tor cell, most object sizes 
remain as unique as in the unpadded case. We observe similar results 
when padding to $256$ bits, the typical block size for AES (not plotted).

\npot (red dotted curves) provides the best anonymity: in the YouTube and 
Ubuntu datasets (Figures~\ref{fig:padme-youtube} and \ref{fig:padme-ubuntu}), there is no single object that remains unique with 
respect to its size; all belong to groups of at least $10$ objects.
We cannot generalize this statement, of course,
as shown by the other two datasets 
(Figures~\ref{fig:padme-files} and \ref{fig:padme-alexa}). 
In general, we see a massive improvement with respect to the unpadded 
case. Recall that this padding scheme is impractically costly,
adding +$100\%$ to the size in the worst case and +$50\%$ in mean. In Table~\ref{table:padme-datasets-overhead}, we see that the mean overhead is of +$45\%$.

Finally, we see the anonymity provided by \padname (green solid 
curves). By design, \padname has an 
acceptable maximum overhead (maximum +$12\%$ and decreasing). 
In three of the four datasets, there is a constant difference 
between our expensive reference point \npot 
and \padname; despite having a decreasing overhead with respect to 
$L$, unlike \npot. This means that although larger files have 
proportionally less protection (\ie less padding in percentage) with 
\padname, this is not critical, as these files are more rare and are harder to 
protect efficiently, even with a na\"ive and costly approach. 
When we observe the percentage of uniquely identifiable objects (objects that 
trivially reveal their plaintext given our perfect adversary), we see a 
significant drop by using \padname: from $83$\% to $3$\% for the Ubuntu
dataset, from $87\%$ to $3\%$ for the Youtube dataset, from $45\%$ to 
$8\%$ for the files dataset and from $68$\% to $6$\% for the Alexa 
dataset. In Table~\ref{table:padme-datasets-overhead}, we see that the 
mean overhead of \padname is around $3\%$, more than an order of 
magnitude smaller than \npot. We also see how using a fixed block size 
can yield high overhead in percentage, in addition to insufficient protection.

\begin{figure}
	\centering
	\begin{subfigure}[t]{.4\textwidth}
		\vspace{-0.2cm}\caption{Dataset `YouTube':}
		\vspace{-0.2cm}\includegraphics[width=\textwidth]{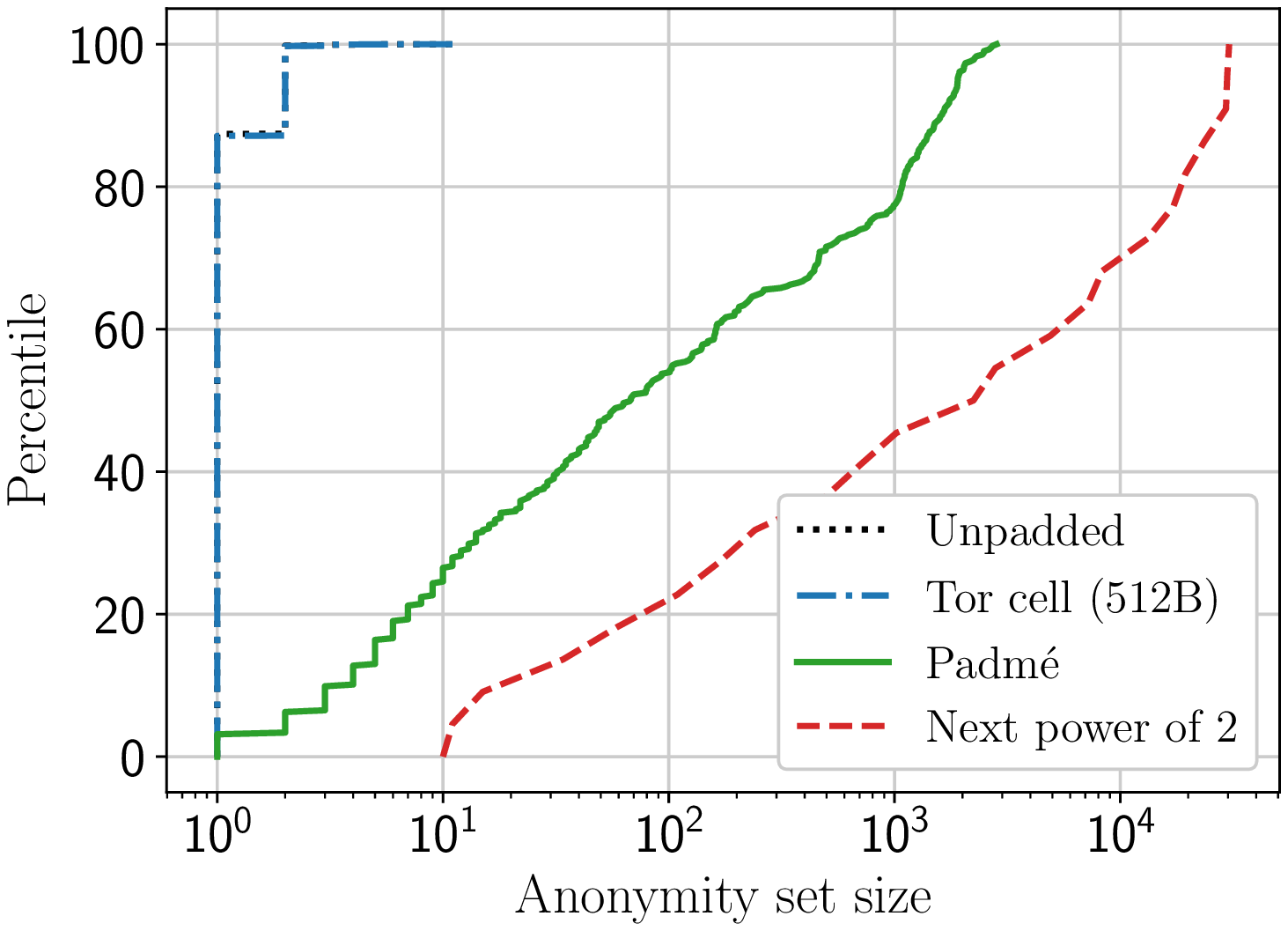}
		\label{fig:padme-youtube}
	\end{subfigure}\hfill
	\begin{subfigure}[t]{.4\textwidth}
		\vspace{-0.6cm}\caption{Dataset `files':}
		\vspace{-0.2cm}\includegraphics[width=\textwidth]{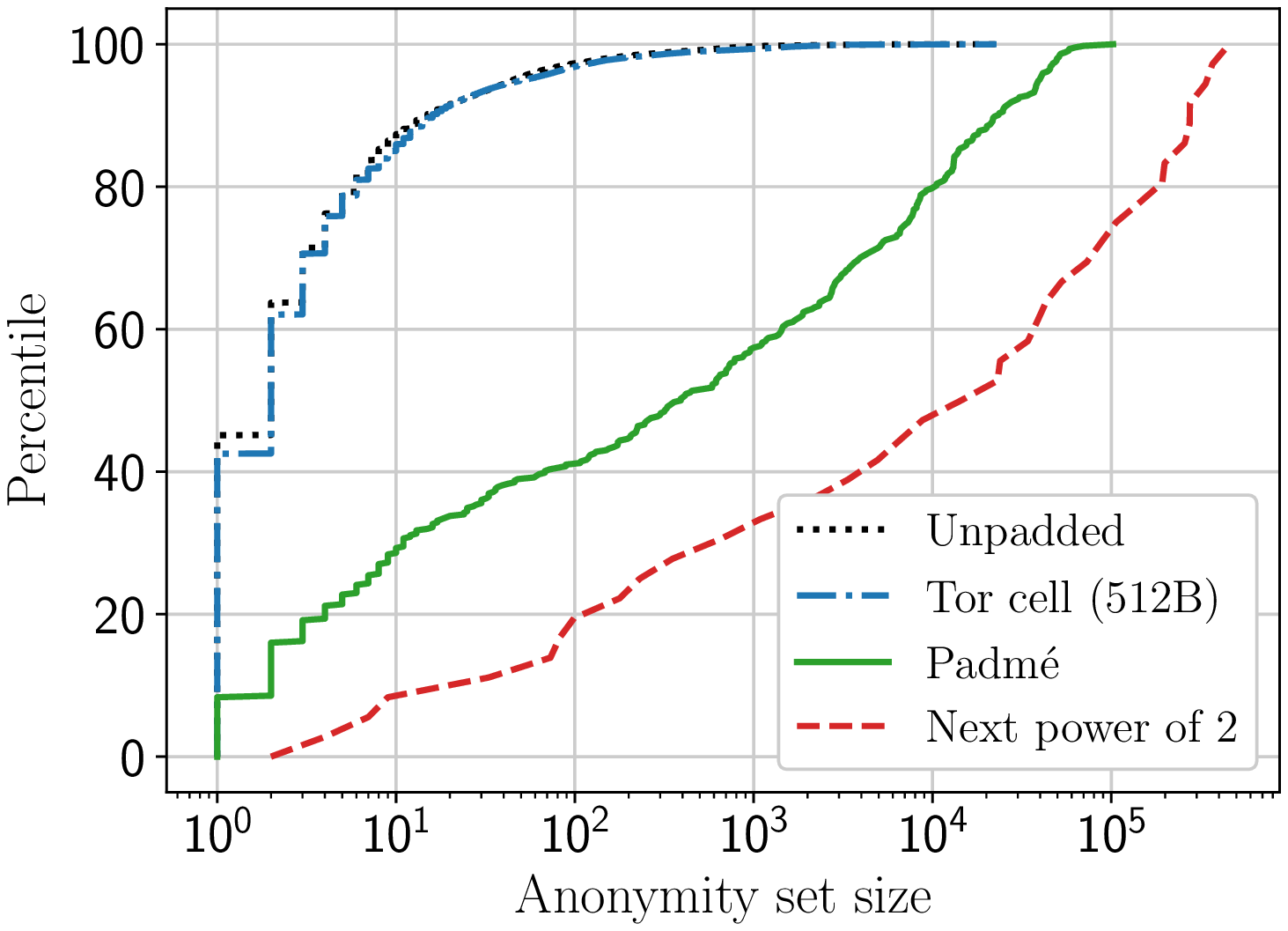}
		\label{fig:padme-files}
	\end{subfigure}\hfill
	\begin{subfigure}[t]{.4\textwidth}
		\vspace{-0.6cm}\caption{Dataset `Ubuntu':}
		\vspace{-0.2cm}\includegraphics[width=\textwidth]{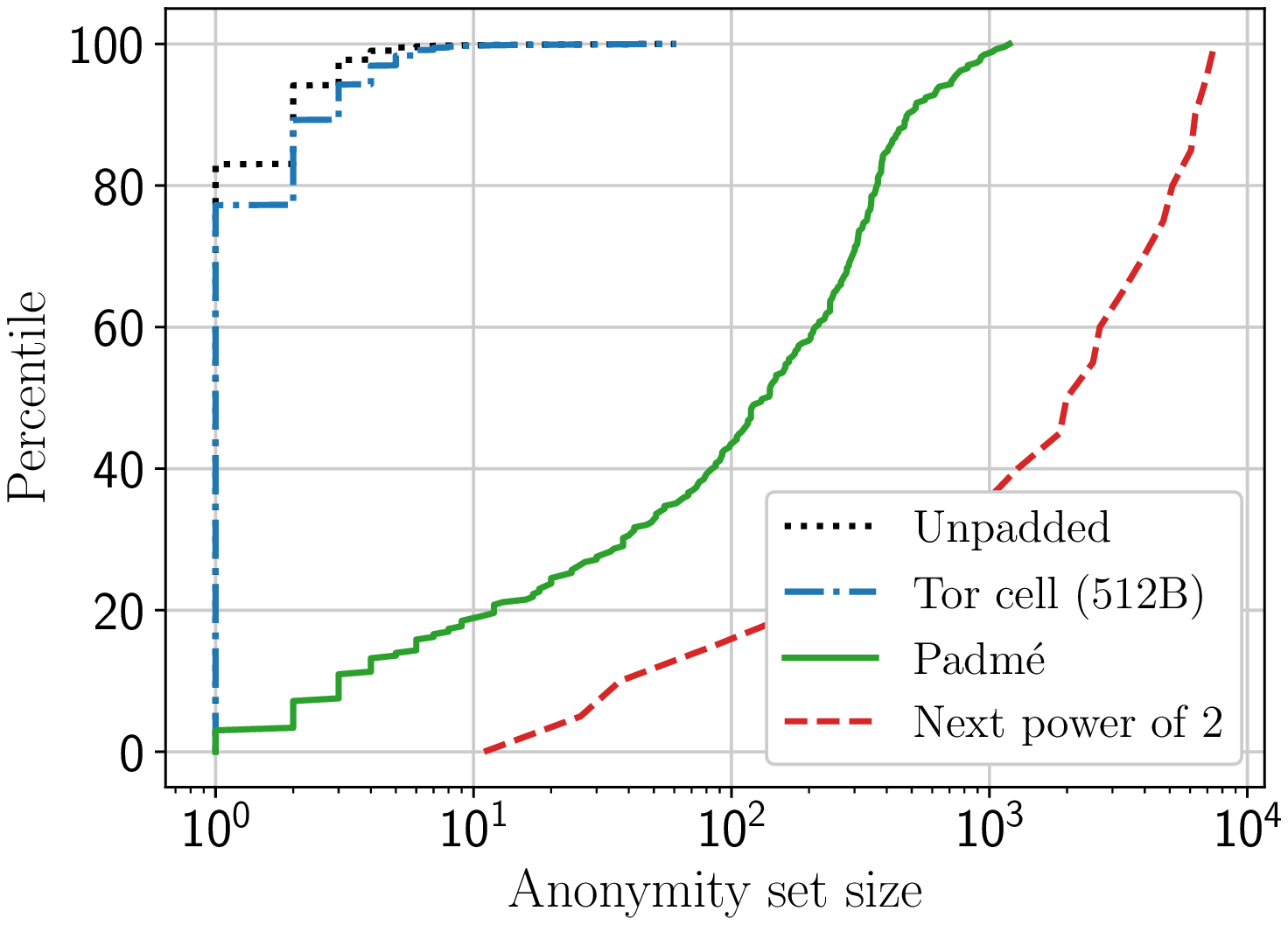}
		\label{fig:padme-ubuntu}
	\end{subfigure}
	\begin{subfigure}[t]{.4\textwidth}
		\vspace{-0.6cm}\caption{Dataset `Alexa':}
		\vspace{-0.2cm}\includegraphics[width=\textwidth]{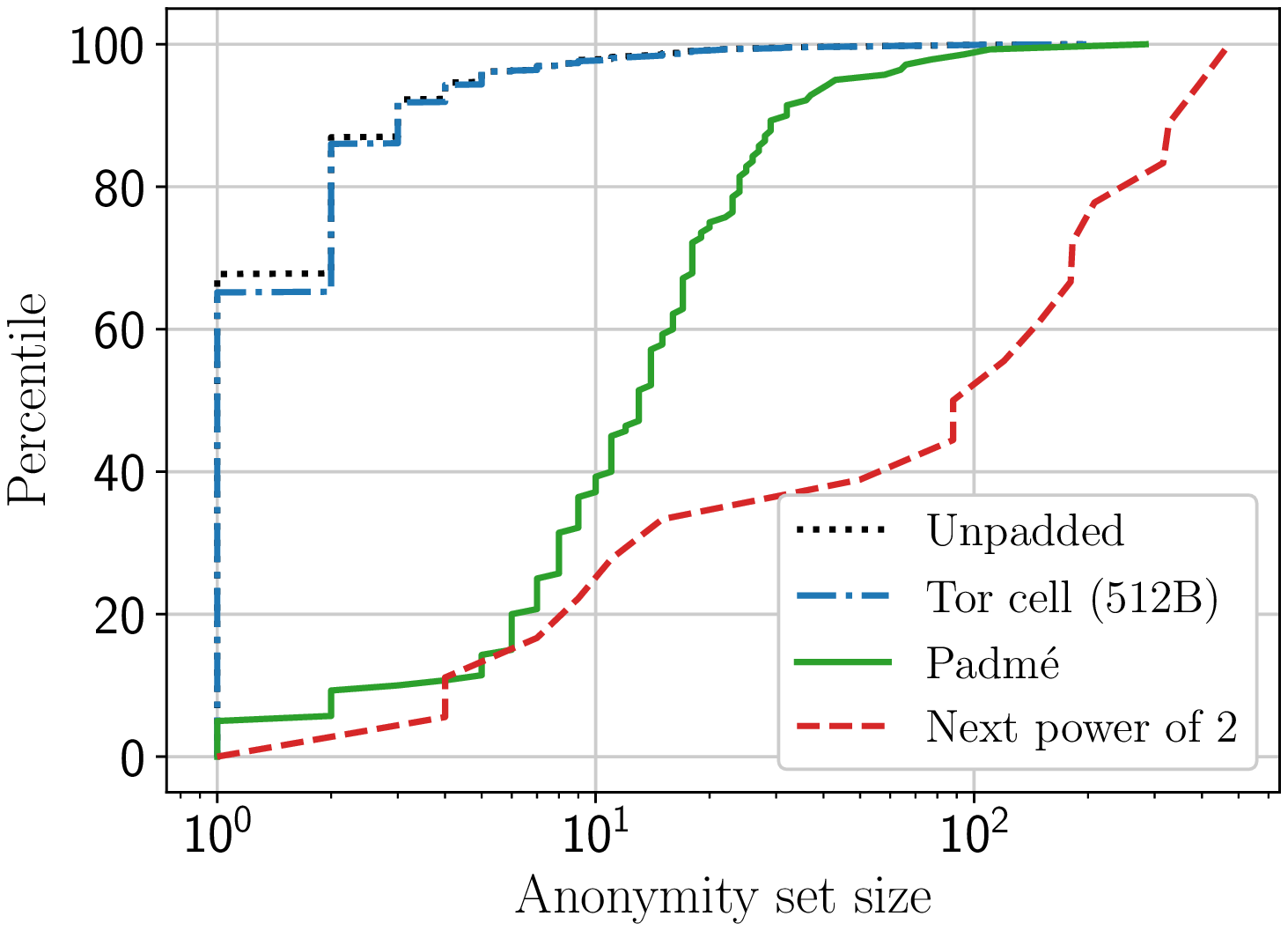}
		\label{fig:padme-alexa}
	\end{subfigure}%
	\vspace{-0.6cm}\caption{Analysis of the anonymity provided by various 
	padding approaches: \npot, \padname, padding with a 
	constant block size and no padding. We 
	measure for each object with how many other objects it becomes 
	indistinguishable after being padded, and plot the distribution. 
	\npot provides better anonymity, at the cost 
	of a drastically higher overhead (at most +$100\%$ instead of +$12\%$). Overheads are shown in Table~\ref{table:padme-datasets-overhead}.}
	\label{fig:padme-datasets}
\end{figure}

\begin{table}
	\caption{Analysis of the overhead, in percentage, of various padding 
	approaches. In the first column, we use $b=512B$ as block size.}
	\vspace{-0.1cm}
	\centering
	\begin{tabular}{llll}
		Dataset & Fixed block size & Next power of 2 & Padmé \\
		\hline
		\normalfont YouTube           & \normalfont 0.01             & \normalfont 44.12           & \normalfont 2.23  \\
		\normalfont files & \normalfont 40.15            & \normalfont 44.18           & \normalfont 3.64  \\
		\normalfont Ubuntu  &\normalfont  14.09            & \normalfont 43.21           &\normalfont  3.12 \\
		\normalfont Alexa               & \normalfont 36.71            & \normalfont 47.12           & \normalfont 3.07  \\
	\end{tabular}
	\label{table:padme-datasets-overhead}
\end{table}

\section{Related Work}
\label{sec:related}

The closest related work \purbs build on
is Broadcast Encryption~\cite{barth06privacy, boneh05collusion, delerablee07identity, fazio12outsider, gentry09adaptive},
which formalizes the security notion behind
a ciphertext for multiple recipients.
In particular, the most relevant notion in (Private) Broadcast Encryption is Recipient Privacy~\cite{barth06privacy}, in which an adversary cannot tell whether a public key is a valid recipient for a given ciphertext. \purbs goes further by enabling multiple simultaneous suites, while achieving indistinguishably from random bits in the \ccatwo model. \purbs also addresses size leakage.

\emph{Traffic morphing}~\cite{wright09trafficmorphing} is a method for 
hiding the traffic of a specific application by masking it as 
traffic of another application and imitating the corresponding 
packet distribution. The tools built upon this method can be 
standalone~\cite{wang2012censorspoofer} or use the concept of Tor 
pluggable 
transport~\cite{moghaddam12skypemorth,weinberg2012stegotorus,winter2013scramblesuit}
that is applied to preventing Tor traffic from being identified and 
censored~\cite{tor16pluggable}. There are two fundamental differences with 
\purbs. First, \purbs focus on a single unit of data; we do not 
yet explore the question of the time distribution of multiple \purbs. 
Second, traffic-morphing systems, in most cases, try to mimic a 
specific transport and sometimes are designed to only hide the traffic of 
one given tool, whereas \purbs are universal and arguably adaptable to any 
underlying application.
Moreover, it has been argued that most traffic-morphing tools 
do not achieve unobservability in real-world settings due to 
discrepancies between their implementations and the systems that they try 
to imitate, because of the uncovered behavior of side protocols, error 
handling, responses to probing, etc.~\cite{houmansadr2013parrot, 
wang2015seeing, frolov19use}.
We believe that for a wide class of applications, using pseudo-random 
uniform blobs, either alone or in combination with other lower-level tools, is 
a potential solution in a different direction.


Traffic analysis aims at inferring the contents of encrypted communication 
by analyzing metadata. The most well-studied application of it is website 
fingerprinting~\cite{panchenko11website,dyer12peek,wang13improved,wang16realistically},
 but it has also been applied to video identification~\cite{reed2016leaky, 
schuster2017beauty, reed2017identifying} and VoIP 
traffic~\cite{wright2007language, chang2008inferring}. In website 
fingerprinting over Tor, research has repeatedly showed that the total 
website size is the feature that helps an adversary the 
most~\cite{cherubin2017website, overdorf2017unique, dyer12peek}. In 
particular, Dyer et al.~\cite{dyer12peek} show the necessity of padding 
the whole website, as opposed to individual packets, to prevent an 
adversary from identifying a website by its observed total size.
They also systematized the existing padding approaches. Wang et 
al.~\cite{wang17walkie} propose deterministic 
and randomized padding strategies tailored for padding Tor traffic against a 
perfect attacker, which inspired our \S\ref{sec:pad}.

Finally, Sphinx~\cite{danezis2009sphinx} is an encrypted packet format for 
mix networks with the goal of minimizing the information revealed to the 
adversary. Sphinx shares similarities with \purbs in its binary format (\eg the 
presence of a group element followed by a ciphertext). 
Unlike \purbs, however, it supports only one cipher suite, and one 
direct recipient (but several nested ones, due to the nature of mix 
networks). To the best of our knowledge, \purbs is the first solution that 
hides all metadata while providing cryptographic agility.

\section{Conclusion}

Conventional encrypted data formats leak information,
via both unencrypted metadata and ciphertext length,
that may be used by attackers to infer sensitive information
via techniques such as traffic analysis and website fingerprinting.
We have argued that this metadata leakage is not necessary,
and as evidence have 
presented \purbs, a generic approach for designing encrypted data formats
that do not leak anything at all, except for the padded length of the
ciphertexts, to anyone without the decryption keys. We have shown that despite
having no cleartext header, \purbs can be efficiently encoded and decoded,
and can simultaneously support 
multiple public keys and cipher suites.
Finally, we have introduced \padname, a padding scheme that reduces the length 
leakage of ciphertexts and has a modest overhead decreasing with file size.
\padme performs significantly better than classic 
padding schemes with fixed block size in terms of anonymity, and its 
overhead is asymptotically lower than using
exponentially increasing padding.


\section*{Acknowledgments}

We are thankful to our anonymous reviewers and our meticulous proof 
shepherd Markulf Kohlweiss for their constructive and thorough feedback 
that has helped us to improve this paper.
We also thank Enis Ceyhun Alp, Cristina Basescu, Kelong Cong, Philipp 
Jovanovic, Apostolos Pyrgelis and Henry Corrigan-Gibbs for their helpful 
comments and suggestions, and Holly B. Cogliati for text editing.
This project was supported in part by grant \#2017-201 of the Strategic 
Focal Area ``Personalized Health and Related Technologies (PHRT)'' of the 
ETH Domain and by grants from the AXA Research Fund, Handshake, and 
the Swiss Data Science Center.

\bibliographystyle{plain}
\bibliography{main}

\appendix
\label{sec:appendix}

\section{Layout}
\label{appendix:algorithms}

Algorithm~\ref{algo:layout} presents
the \layout algorithm a sender uses in step (8) 
of \penc.
\layout arranges a \purb's components in a 
continuous byte array.

\para{Notation.}
We denote by $a[i:j] \gets b$, the operation of copying the bits 
of b at the positions $a[i], a[i+1], \cdots a[j-1]$. When written like this, 
$b$ always has correct length of $j-i$ bits, and we assume $i<j$. If, before 
an operation $a[i:j]\gets b$, $|a|<j$, we first grow $a$ to length $j$. We 
sometimes write $a[i:] \gets b$ instead of $a[i:|b|] \gets b$.
We use a ``reservation array'', which is an 
array with a method array.isFree(start,end) that returns True if and only if 
none of the bits $\text{array}[i], \text{array}[i+1],\cdots\text{array}[j-1]$ were 
previously assigned a value, and False otherwise.

\begin{algorithm*}
	\begin{multicols}{2}
	
	\LinesNumbered
	\SetAlgoLined
	
	\fontsize{8pt}{10pt}\selectfont
	
	\SetKwInOut{Input}{Input}
	\SetKwInOut{Output}{Output}
	
	\tcp{$\tau_i$ is an encoded public key of a suite $\suite_i$}
	\tcp{$\keys_i = \langle Z_1, \ldots, Z_r \rangle$ are entry-point keys}
	\tcp{$\aux_i = \langle P_1, \ldots, P_r \rangle$ are entry-point positions}
	\tcp{SuiteAllowedPositions are public values}
	\Input{$\langle \tau_1, \ldots, \tau_n \rangle$, $\langle \keys_1, 
	\ldots, \keys_n \rangle$, $\langle \aux_1, \ldots, \aux_n \rangle$, 
	$\langle \suite_1, \ldots, \suite_n \rangle$, $K$, $\meta$, 
	$\ctxtpayload$, SuiteAllowedPositions}
	\Output{byte[]}
	\BlankLine

	\tcp{determine public-key positions for each suite}
	layout = []\tcp*[l]{public-key and entry-point assignments}
	pubkey\_pos = []\tcp*[l]{chosen primary position per suite}
	pubkey\_fixed = []\tcp*[l]{all positions fixed so far}
	\ForEach{$\tau_i$ \normalfont{in} $\langle \tau_1, \ldots, \tau_n \rangle$}{
		\tcp{decide suite's primary public key position}
		\For{\normalfont{pos} $\in$ SuiteAllowedPositions($\suite_i$)}{
			\If{\normalfont{pubkey\_fixed.isFree}(pos.start, pos.end)}{
				pubkey\_pos.append($\langle \tau_i, \mathrm{pos}\rangle$)\;
				layout[pos.start:pos.end] $\gets$ $\tau_i$\;
				\textbf{break}\;
			}
		}
		\tcp{later suites cannot modify these positions}
		\tcp{without disrupting this suite's XOR}
		\For{\normalfont{pos} $\in$ SuiteAllowedPositions($\suite_i$)}{
			pubkey\_fixed[pos.start:pos.end] $\gets$ `\texttt{F}'\;
		}
	}
	\BlankLine
	
	\tcp{reserve entry-point positions in hash tables}
	entrypoints = []\;
	\ForEach{\normalfont{$\aux_i$ in $\langle \aux_1, \ldots, \aux_n \rangle$}}{
		\While{\normalfont{$\aux_i$ \textbf{not} empty}}{
			$P \gets \aux_i.\text{pop()}$\;
			ht\_len = 1\tcp*[l]{length of current hash table}
			ht\_pos = 0\tcp*[l]{position of this hash table}
			\While{True}{
				
				index = $P \mod \mathrm{ht\_len}$\tcp*[l]{selected entry}
				start = ht\_pos + index * entrypoint\_len\;
				end = start + entrypoint\_len\;
				
				\If{\normalfont{layout.isFree}(start, end)}{
					layout[start:end] $\getrand  \{0, 1\}^{\text{end-start}}$\;
					entrypoints.append($\langle$\text{start, end}, $\suite_i \rangle$)\;
					\textbf{break}\;
				}
				\tcp{if not free, double table size}
				ht\_pos += ht\_len * entrypoint\_len\;
				ht\_len *=~2\;
			}
		}
	}
	\BlankLine
	\vfill\null
	\columnbreak
	
	\tcp{fill empty space in the layout with random bits}
	\ForEach{\normalfont{start, end < layout.end} 
	}{\If{\normalfont{layout.isFree}(start, 
	end)}{layout[start:end] $\getrand \{0, 1\}^{\text{end-start}}$ }}
	\BlankLine
	
	\tcp{place the payload just past the header layout}
	\meta.payload\_start = |layout|\;
	\meta.payload\_end = |layout| + |$\ctxtpayload$|\;
	\BlankLine

	\tcp{fill entry-point reservations with ciphertexts}
	\ForEach{\normalfont{$\keys_i$ in $\langle \keys_1, \ldots, \keys_n 
	\rangle$}}{
		\While{\normalfont{$\keys_i$ \textbf{not} empty}}{
			$Z = \keys_i.\text{pop()}$\;
			$\langle \mathrm{start}, \mathrm{end}, \suite \rangle \gets$ entrypoints.pop()\;
			\tcp{Encrypt an entry point}
			$e \gets \E_Z(K \parallel \meta)$\;
			layout[start:end] $\gets e$\;
		}
	}
	
	\tcp{compute the padding and append it to layout}
	purb\_len $\gets$ \padname(|layout| + |$\ctxtpayload$| + mac\_len)\;
	mac\_pos $\gets$ purb\_len - mac\_len\;
	\While{\normalfont{\textbf{not} pubkey\_fixed.isFree(mac\_pos, purb\_len)}}{
		\tcp{MAC mustn't overlap public-key positions:}
		\tcp{if so, we pad to the next \padname size}
		purb\_len $\gets$ \padname(purb\_len + 1)\;
		mac\_pos $\gets$ purb\_len - mac\_len\;
	}

	padding\_len $\gets$ mac\_pos - \meta.payload\_end\;
	padding $\getrand \{0, 1\}^\text{padding\_len}$\tcp*[l]{random padding}
	layout.append($\ctxtpayload \parallel$ padding)\;
	\BlankLine

	\tcp{XOR suites' public key positions into primary}
	\For{\normalfont{($\tau_i$, pos)} $\in$ \normalfont{pubkey\_pos}}{
		buffer = $\tau_i$\;
		\For{\normalfont{altpos} $\in$ 
		\normalfont{SuiteAllowedPositions($\suite_i$)}}{
			buffer = buffer $\oplus$ layout[altpos.start~:~altpos.end]\;
		}
		layout[pos.start:pos.end] $\gets$ buffer\;
		\tcp{now $\bigoplus$ SuiteAllowedPositions($\suite_i$) = $\tau_i$}
	}
	\BlankLine
	
	\Return layout

	\caption{\layout}
	\label{algo:layout}
\end{multicols}
\end{algorithm*}

\section{Positions for Public Keys}
\label{appendix:allowed-positions}

This section provides an example of possible sets of allowed 
public key positions for the suites in the~\purb encoding. We emphasize that finding 
an optimal set of positions was not the focus of this work. The intention is 
merely to show that such sets exist and to offer a concrete example (which is used 
for the compactness experiment, Figure~\ref{fig:compactness}).

\para{Example.} We use the required and recommended suites in the latest 
draft 
of TLS 1.3~\cite{rfc8446} as an example of suites a \purb could 
theoretically support. 
The suites and groups are shown in Table~\ref{table:tls-suites}.

The \purb concept of ``suite'' combines both ``suite'' and 
``group'' in TLS. For instance, a PURB suite could be 
PURB\_AES\_128\_GCM\-\_SHA\_256\-\_SECP256R1. We show possible 
\purb suites in Table~\ref{table:purb-suites}. For the sake of 
simplicity, we introduce aliases in the table, and will further refer to those 
suites as suite A-F. In Table~\ref{table:possible-positions1}, we show a 
possible assignment. For instance, if only suites A and C are used, the 
public key for A would be placed in $[0,64]$, while value in $[96,160]$ is 
changed so that the XOR of $[0,64]$ and $[96,160]$ equals the key for B.
Note that a sender must respect the suite order A-F during encoding.
We provide a simple python script to design such sets in the code 
repository.

\begin{table}
	\caption{Suites and groups described in the latest draft of TLS 1.3.}
	\centering
	\begin{tabular}{ll}
		\textbf{Symmetric/Hash Algorithms}   &             \\
		\hline 
		{\normalfont TLS\_AES\_128\_GCM\_SHA256}      & {\normalfont 
		Required}    \\
		{\normalfont TLS\_AES\_256\_GCM\_SHA384}      & {\normalfont 
		Recomm} \\
		{\normalfont TLS\_CHACHA20\_POLY1305\_SHA256} & {\normalfont 
		Recomm} \\
		{\normalfont TLS\_AES\_128\_CCM\_SHA256}      & {\normalfont 
		Optional}    \\
		{\normalfont TLS\_AES\_128\_CCM\_8\_SHA256}   & {\normalfont 
		Optional}    \\
	\hline 
		\textbf{Key Exchange Groups}    &             \\
		\hline 
		{\normalfont secp256r1}                       & {\normalfont 
		Required}    \\
		{\normalfont x25519}                          & {\normalfont 
		Recomm} \\
		{\normalfont secp384r1}                       & {\normalfont 
		Optional}    \\
		{\normalfont secp521r1}                       & {\normalfont 
		Optional}    \\
		{\normalfont x448}                            & {\normalfont 
		Optional}    \\
		{\normalfont ffdhe2048}                       & {\normalfont 
		Optional}    \\
		{\normalfont ffdhe3072}                       & {\normalfont 
		Optional}    \\
		{\normalfont ffdhe4096}                       & {\normalfont 
		Optional}    \\
		{\normalfont ffdhe6144}                       & {\normalfont 
		Optional}    \\
		{\normalfont ffdhe8192}                       & {\normalfont 
		Optional}    \\
	\end{tabular}
	\label{table:tls-suites}
\end{table}

\begin{table}
	\caption{Example of Allowed Positions per suite. Here, the algorithm 
		simply finds any mapping so that each suite can coexist in a PURB. 
		The 
		receiver must XOR the values at all possible positions of a suite to 
		obtain an~encoded public key..}
	\centering
	\begin{tabular}{ll}
		\textbf{Suite}   &  \text{Possible positions}      \\
		\hline 
		{\normalfont A}      & {\normalfont $\{0\}$ }   \\
		{\normalfont B}      & {\normalfont $\{0, 64\}$ }   \\
		{\normalfont C}      & {\normalfont $\{0, 96\}$ }   \\
		{\normalfont D}      & {\normalfont $\{0, 32, 64, 160\}$ }   \\
		{\normalfont E}      & {\normalfont $\{0, 64, 128, 192\}$ }   \\
		{\normalfont F}      & {\normalfont $\{0, 32, 64, 96, 128, 256\}$ }   \\
	\end{tabular}
	\label{table:possible-positions1}
\end{table}

\begin{table*}[h]
	\caption{PURB Suites. ``Suite A'' is a shorthand for the first suite.}
	\centering
	\begin{tabular}{llll}
		\textbf{Alias}  & \textbf{PURB 
		Suite}                                    
		& 
		\textbf{Public key {[}B{]}} & \textbf{EntryPoint {[}B{]}} \\
		\hline 
		{\normalfont A} & {\normalfont 
		PURB\_AES\_128\_GCM\_SHA\_256\_SECP256R1}      & {\normalfont 
		64}                    & {\normalfont 48}                        \\
		{\normalfont B} & {\normalfont 
		PURB\_AES\_128\_GCM\_SHA\_256\_X25519}         & {\normalfont 
		32}                    & {\normalfont 48}                        \\
		{\normalfont C} & {\normalfont 
		PURB\_AES\_256\_GCM\_SHA\_384\_SECP256R1}      & {\normalfont 
		64}                    & {\normalfont 80}                        \\
		{\normalfont D} & {\normalfont 
		PURB\_AES\_256\_GCM\_SHA\_384\_X25519}         & {\normalfont 
		32}                    & {\normalfont 80}                        \\
		{\normalfont E} & {\normalfont 
		PURB\_CHACHA20\_POLY1305\_SHA\_256\_SECP256R1} & 
		{\normalfont 
		64}                    & {\normalfont 64}                        \\
		{\normalfont F} & {\normalfont 
		PURB\_CHACHA20\_POLY1305\_SHA\_256\_X25519}    & 
		{\normalfont 
		32}                    & {\normalfont 64}        \\               
	\end{tabular}
	\label{table:purb-suites}
\end{table*}

\section{Default Schemes for Payload}

In addition to \purb suites, a list of suitable candidates for a payload 
encryption scheme $(\enc, \dec)$, a MAC algorithm \mac, and a hash 
function $\hashpayload$ must be determined and standardized.
This list can be seamlessly updated with time, as an encoder makes the 
choice and records it in \meta on per-\purb basis.
The chosen schemes are shared by all the suites included in the \purb, 
hence these schemes must match the security level of the suite with 
the highest bit-wise security.
An example of suitable candidates, given the suites from 
Table~\ref{table:purb-suites}, is $(\enc, \dec) = \text{AES256-CBC}$, $\mac 
= \text{HMAC-SHA384}$, and $\hashpayload = \text{SHA3-384}$.

\com{	This seems so trivial to me that it doesn't add much,
	and it would have to be revised to be clearer and consistent
	with the mainline text, which there's no time for. -baf
\section{Leakage: Next Power of 2}


As mentioned in the body of the paper, we prove here that padding to the 
next 
power of two yields a leakage and a overhead in $O(\log \log L)$.

\para{\emph{Lemma 1:}} Padding with buckets $b_i$ following $2^i$ (or, 
equivalently, padding using $f(L) = 2^{\ceil{\log_2L}}$) yield an overhead of  
$O(\log \log L)$.
	
\para{\emph{Proof 1:}} The information we are leaking is in which bucket 
the plaintext landed. To compute how many bits of information this 
represents, we compute the number of buckets $n$ needed to 
accommodate contents of size up to $L$. Since the buckets' sizes are 
powers of two, dictated by $b_i = 2^i$, we obtain:
$$\sum_{i=0}^n 2^i \ge L$$
$$2^{n+1} \ge L$$
$$n \ge \log_2(L-1) - 1$$
$$n \approx O(\log L)$$
Hence, there are $n$ buckets in the range $[0;L]$. Since we leak the 
number $n$, we leak $\log_2(n)$ or $O(\log \log L)$ bits. The same 
reasoning can be used to show the overhead, also in $O(\log \log L)$.
}

\section{Security Proofs}

This section contains the proofs of the security properties provided by 
\mspurb.

\subsection{Preliminaries}
\label{sec:prelim}

Before diving into proving the security of our scheme, we define what it 
means to be \indcca- and \ccatwo-secure for the primitives that \mspurb 
builds upon.

\para{Key-Encapsulation Mechanism (KEM).}
Following the definition from Katz \& Lindell~\cite{katz14introduction}, we 
begin by defining 
KEM as a tuple of PPT algorithms.

\begin{syntax}[KEM]
	\label{syntax:kem}
	\hangtwo
	$\kem\setup(\onelambda) \to \suite$: Given a security parameter 
	$\lambda$, initialize a cipher suite \suite.
	
	\hangone
	$\kem\keygen(\suite) \to (sk, pk)$: Given a cipher suite \suite, 
	generate a (private, public) key pair.
	
	\hangone
	$\kem\encap(pk) \to (c, k)$: Given a public key $pk$, output a ciphertext 
	$c$ and a key $k$.
	
	\hangone
	$\kem\decap(sk, c) \to k/\bot$: Given a private key $sk$ and a ciphertext $c$, 
	output a key $k$ or a special symbol $\bot$ denoting failure.
\end{syntax}

Consider an \indcca security game against an adaptive adversary \A:
\begin{game}[KEM]
	\label{game:kem}
  The KEM \indcca game for a security parameter $\lambda$ is between a
  challenger and an adaptive adversary \A. It proceeds along the
  following phases.

	\hangone
	\textbf{Init:} The challenger and adversary take $\lambda$ as input.
  The adversary outputs a cipher suite $\suite$ it wants to attack.
  The challenger verifies that $\suite$ is a valid cipher suite, i.e.,
  that it a valid output of $\kem\setup(\onelambda)$. The challenger aborts, and
  sets $b^\star \getrand \{0, 1\}$ if $\suite$ is not valid.
	
	\hangone
	\textbf{Setup:} The challenger runs $(sk, pk) \gets 
	\kem\keygen(\suite)$ and gives $pk$ to \A.
	
	\hangone
	\textbf{Phase 1:} \A can make
	decapsulation queries $\qdecap(c)$ with
  ciphertexts $c$ of its choice,
	to the challenger who responds with $\kem\decap(sk, c)$.
	
	\hangone
	\textbf{Challenge:} The challenger runs $(c^\star, k_0) \gets \kem\encap(pk)$ 
	and generates $k_1 \getrand \{0,1\}^{|k_0|}$.
	The challenger picks $ b \getrand \{0,1\}$ and sends $\langle 
	c^\star, k_b \rangle$ to~\A.
	
	\hangone
	\textbf{Phase 2:} \A continues querying $\qdecap(c)$
	with the restriction that $c \neq c^\star$.
	
	\hangone
	\textbf{Guess:} \A outputs its guess $b^\star$ for $b$ and 
	wins if $b^\star = b$.
\end{game}

\newcommand{\advantage}{\textsf{Adv}}
\newcommand{\advkemcca}[1][\A]{\advantage^{\textsf{cca2}}_{\textsf{KEM},#1}}
\newcommand{\advae}[1][\A]{\advantage^{\textsf{ind\$-cca2}}_{\Pi,#1}}
\newcommand{\advmsbe}[1][\A]{\advantage^{\textsf{cca2-out}}_{\textsf{msbe},#1}}
\newcommand{\advmsbein}[1][\A]{\advantage^{\textsf{cpa-in}}_{\textsf{msbe},#1}}
\newcommand{\advmacsforge}[1][\A]{\advantage^{\textsf{suf}}_{\mac,#1}}
\newcommand{\advmacind}[1][\A]{\advantage^{\textsf{ind\$}}_{\mac,#1}}
\newcommand{\advindcpa}[1][\A]{\advantage^{\textsf{ind\$-cpa}}_{(\enc,\dec),#1}}

We define \A's advantage in this game as:
\begin{equation*}
  \advkemcca(\onelambda) = 2 \left| \textrm{Pr}[b = b^{\star}] - \tfrac{1}{2} \right|.
\end{equation*}
We say that a KEM is \indcca-secure if $\advkemcca(\onelambda)$ is negligible in
the security parameter.

\begin{definition}
  We that a KEM is \emph{perfectly correct} if for all $(sk, pk) \gets
  \kem\keygen(\suite)$ and for all $(c, k) \gets \kem\encap(pk)$ we have $k =
  \kem\decap(sk, c)$.
\end{definition}

\begin{insta}[IES-KEM]
	\label{insta:kem}
	We instantiate a KEM based on the Integrated Encryption 
	Scheme~\cite{abdalla01oracle} (see \S\ref{sec:ies} for details).
	\vspace{0.1cm}
	
	\hangone
	$\ies\setup(\onelambda)$: Initialize a cipher suite \suite = 
	$\langle \G, p, g, \hasha \rangle$, where $\G$ is a cyclic group of order 
	$p$ and generated by $g$, and $\hasha:\G \to \{0, 1\}^{2\lambda}$ is a 
	hash function.
	
	\hangone
	$\ies\keygen(\suite)$: Pick $x \in \mathbb{Z}_p$, compute 
	$X = g^{x}$, and output $(sk = x, pk = X)$.
	
	\hangone
	$\ies\encap(pk)$: Given $pk = Y$, pick $x \in \mathbb{Z}_p$, compute 
	$X = g^{x}$, and output $\langle c = X, k = \hasha(Y^x) \rangle$.
	
	\hangone
	$\ies\decap(sk, c)$: Given $sk = y$ and $c = X$, 
	output a key $k = \hasha(X^y)$.
\end{insta}

\begin{theorem}[Theorem 11.22~\cite{katz14introduction} and Section 
7~\cite{abdalla01oracle}]
	If the gap-CDH problem is hard relative to $\G$, and $\hasha$ is 
	modeled as a random oracle, then IES-KEM is an 
	\textnormal{\indcca}-secure KEM.
\end{theorem}

\para{Multi-Suite Broadcast Encryption.}
We consider \mspurb as a multi-suite broadcast encryption (MSBE) 
scheme 
extending the single-suite setting by Barth et al.~\cite{boneh05collusion}.

\begin{syntax}[MSBE]
	\label{syntax:msbe}
	\hangtwo
	$\msbe\setup(\onelambda) \to \suite$: Given a security parameter 
	$\lambda$, initialize a cipher suite $\suite$.
	
	\hangone
	$\msbe\keygen(\suite) \to (sk, pk)$: Given a cipher suite \suite, 
	generate a (private, public) key pair.
	
	\hangone
	$\msbe\enc(R, m) \to c$: Given a set of public keys $R = \{pk_1,\ldots, 
	pk_r\}$ with corresponding cipher suites $\suite_1, \ldots, \suite_r$ and a 
	message $m$, generate a ciphertext $c$.
	
	\hangone
	$\msbe\dec(sk, c) \to m/\bot$:  Given a private key $sk$ and the 
	ciphertext $c$, return a message $m$ or $\bot$ if $c$ does not 
	decrypt correctly.
\end{syntax}
Note that \mspurb as described in \S\ref{sec:purb-complete} satisfies the
syntax of a multi-suite broadcast encryption scheme.

Barth et al.~\cite{barth06privacy} define the security of broadcast 
encryption schemes under adaptive chosen-chiphertext attack for 
single-suite schemes. 
Here, we adjust this definition to the multi-suite setting, and 
instead require that the ciphertext is indistinguishable from a random string 
(\ccatwo).

\begin{game}[MSBE]
	\label{game:msbe}
  The MSBE \ccatwo game for a security parameter $\lambda$ is between a
  challenger and an adversary \A.
  It proceeds along the following phases.

	\hangone
	\textbf{Init:} The challenger and adversary take $\lambda$ as input.
  The adversary outputs
  a number of recipients $r$ and corresponding cipher suites
  $\suite_1, \ldots, \suite_r$ it wants to attack. Let $s$ be the number of
  unique cipher suites.
  The challenger verifies,
  for each $i \in \{1, \ldots, r\}$,
  that $\suite_i$ is a valid cipher suite, i.e.,
  that it is a valid output of $\msbe\setup(\onelambda)$. The challenger aborts, and
  sets $b^\star \getrand \{0, 1\}$ if the suites are not all valid.
	
	\hangone
	\textbf{Setup:} The challenger generates private-public key 
	pairs for each recipient $i$ given by \A by running $(sk_i, pk_i) \gets 
	\msbe\keygen(\suite_i)$ and gives $R = \{pk_1,\ldots, pk_r\}$ to \A.
	
	\hangone
	\textbf{Phase 1:} \A
  can make decryption queries 
	$\qdec(pk_i, c)$ to the challenger for any $pk_i \in R$ and any ciphertext 
	$c$ of its choice.
	The challenger replies with $\msbe\dec(sk_i, c)$.
	
	\hangone
	\textbf{Challenge:} \A outputs $m^\star$. The challenger generates 
	$c_0 = \msbe\enc(R, m^\star)$ and $c_1 \getrand \{0, 1\}^{|c_0|}$.
	The challenger picks $ b \getrand \{0,1\}$ and sends $c^\star = 
	c_b$ to \A.
	
	\hangone
	\textbf{Phase 2:} \A continues making decryption queries $\qdec(pk_i, c)$
	with a restriction that $c \neq c^\star$.
	
	\hangone
	\textbf{Guess:} \A outputs its guess $b^\star$ for $b$ and 
	wins if $b^\star = b$.
\end{game}
We define \A's advantage in this game as:
\begin{equation*}
  \advmsbe(\onelambda) = 2\left| \textrm{Pr}[b = b^{\star}] - \tfrac{1}{2} \right|.
\end{equation*}
We say that a MSBE scheme is \ccatwo-secure if $\advmsbe(\onelambda)$ is negligible in
the security parameter.

\medskip

Finally, we require that the $\mac$ scheme
is strongly unforgeable under an adaptive chosen-message attack \emph{and}  
outputs tags that are indistinguishable from random.
A $\mac$ scheme is given by the algorithms $\mac.\keygen, \M,$ and $\V$, where
$\mac.\keygen(\onelambda)$ outputs a key $\mackey$. To compute a tag on the
message $m$, run $\sigma = \M_{\mackey}(m)$. The verification algorithm
$\V_{\mackey}(m, \sigma)$ outputs $\top$ if $\sigma$ is a valid tag on the
message $m$ and $\bot$ otherwise.
We formalize the strong unforgeability and indistinguishability properties
using the following simple games.

\begin{game}[MAC-sforge]
	The MAC-sforge game for a security parameter $\lambda$
  is between a challenger and an adversary \A.
	
	\hangone
	\textbf{Setup:} The challenger and adversary take $\lambda$ as input.
  The challenger generates a MAC key $\mackey \gets \mac.\keygen(1^\lambda)$.
	
	\hangone
	\textbf{Challenge:} The adversary \A is given oracle access to the oracles
  $\M(\cdot)$ and $\V(\cdot)$.
  On a query $\M(m)$ the challenger returns $\sigma = \M_{\mackey}(m)$.
  On a query $\V(m, \sigma)$ the challenger returns $\V_{\mackey}(m, \sigma)$.
	
	\hangone
	\textbf{Output:} \A eventually outputs a message-tag pair $(m, \sigma)$. \A wins if 
	$\V_{\mackey}(m, \sigma) = 1$ \emph{and} \A has not made a query $\M(m)$ that
  returned $\sigma$.
\end{game}
We define \A's advantage in this game as:
\begin{equation*}
  \advmacsforge(\onelambda) = \textrm{Pr}[\text{$\A$ wins}].
\end{equation*}
We say that a MAC scheme is strongly unforgeable under adaptive chosen-message attacks if $\advmacsforge(\onelambda)$ is negligible in the security parameter.

\begin{game}[MAC-IND\$]
  The MAC-IND\$ game is between a challenger and an adversary \A.

  \hangone
	\textbf{Setup:} The challenger and adversary take $\lambda$ as input.
  The challenger generates a MAC key $\mackey \gets \mac.\keygen(1^\lambda)$ and
  picks a bit $b \getrand \{0, 1\}$.

  \hangone
  \textbf{Challenge:} The adversary outputs a message $m$. The challenger
  computes $\sigma_0 = \M_{\mackey}(m)$ and $\sigma_1~\getrand~\{0,
  1\}^{|\sigma_0|}$ and returns $\sigma_b$.

  \hangone
  \textbf{Output:} The adversary outputs its guess $b^\star$ of $b$, and wins if
  $b^\star = b$. 
\end{game}
We define \A's advantage in this game as:
\begin{equation*}
  \advmacind(\onelambda) = 2 \left|  \textrm{Pr}[b = b^\star] - \tfrac{1}{2} \right|.
\end{equation*}
We say that the tags of a MAC scheme are indistinguishable from random if $\advmacind(\onelambda)$ is negligible in the security parameter.

\subsection{Proof of Theorem~\ref{theorem:outsider}}
\label{proof:outsider}

We prove the \ccatwo security of \mspurb as an~MSBE scheme. 
More precisely, we will show that there 
exists adversaries $\B_1, \ldots, \B_5$ such that
\begin{align*}
\advmsbe(\onelambda) \leq
  \;& r \left( \advkemcca[\B_1](\onelambda) + \advae[\B_2](\onelambda)\right) + \\
  &\advmacsforge[\B_3](\onelambda) +
  \advmacind[\B_4](\onelambda) + \\
  &\advindcpa[\B_5](\onelambda).
\end{align*}
Thus, given our assumptions, $\advmsbe(\onelambda)$ is indeed negligible in
$\lambda$. To do so we use a sequence of games.
This sequence of games step by step transforms from 
the situation where $b = 0$ in the \ccatwo game of MSBE, \ie the adversary 
receives the real ciphertext, to $b = 1$, \ie the adversary receives a random 
string.

\begin{game}[\gamestart]
	\label{game:start}
  This game is as the original MSBE \ccatwo game where $b = 0$.\\
\end{game}

\begin{game}[\gamememkeys]
	\label{game:memkeys}
  As in \gamestart, but the challenger will no longer call \hdr\decap to derive
  the keys $k_i$ on ciphertexts derived from the challenge
  ciphertext $c^\star$. In particular, for every recipient $pk_i$ using a suite
  $\suite_j$, we store $(X_j^\star, k_i^\star)$ when constructing the PURB headers for
  the challenge ciphertext. 
  Then, when receiving a decryption query for a recipient 
  $\qdec(pk_i(\suite_j), c)$, we proceed by
  following \pdec. 
  If the encoded public key
  $\tau$ recovered in step (1) of \pdec is such that $\unhide(\tau) =
  X_j^{\star}$, then we use $k_i = k_i^\star$ (as stored when creating the
  challenge ciphertext) directly,
  rather than computing $k_i =
  \hdr\decap(y_i, \tau)$ in step (3) of $\pdec$. If the encoded public key 
  $\tau$ does not
  match $X_j^{\star}$, then the challenger proceeds as before.
\end{game}

\begin{game}[\gamesubkeys]
	\label{game:subkeys}
	As in \gamememkeys, but we change how the keys $k_1^\star,\ldots, k_r^\star$
  for the \emph{challenge ciphertext}
  are computed in $\hdr\encap$.
  Rather than computing $k_i^\star =
  \hasha(Y_i^{x})$ as in step (2) of \hdr\encap, we set $k_i^\star \getrand \{0,
  1\}^{\lambda_H}$ for all the keys, where $\lambda_H$ is the bit-length of the
  corresponding hash function $\hasha$.
  Recall that as per the changes in \gamememkeys, the challenger will store 
  $k_i^\star$ generated in this way, and use them directly (without calling
  \hdr\decap) when asked to decrypt variants of the challenge ciphertext.
\end{game}

\begin{game}[\gamememmsg]
	\label{game:memmsg}
  Let $e_i$ be the encrypted entry point under key $Z_i$ (derived from 
  $k_i$) for recipient $i$ 
  computed in line 47 of \layout (step (8) of \penc). The game goes as in 
  \gamesubkeys, but for the challenge ciphertext, the challenger saves the 
	mapping of the challenge entry points and the encapsulated key 
	$K^\star$ 
	with metadata $\meta^\star$: $(e^\star_i, k_i^\star, K^\star \parallel 
	\meta^\star)$.
	If the challenger receives a decryption query $\qdec(pk_i(\suite_i), c)$ it 
  proceeds as before, except when it should decrypt $e^\star_i$ using key 
  $k_i^\star$ in
  step (4) of \pdec. In that case, it acts as if the decryption returned $K^\star
  \parallel \meta^\star$.
\end{game}

\begin{game}[\gamesubenc]
	\label{game:subenc}
	As in \gamememmsg, but the challenger replaces $ e^\star_1,\ldots, 
	e^\star_r$ in the challenge ciphertext with 
	random strings of the appropriate length. Note that per the change in 
	\gamememmsg, the challenger will not try to decrypt these $e^\star_i$, 
	but will recover $K^\star$ and $\meta^\star$ directly instead. 
\end{game}

\begin{game}[\gamerepbot]
	\label{game:repbot}
	As in \gamesubenc, but the challenger replies differently to the queries 
	$\qdec(pk_i(\suite_i), c)$ where $c$ is not 
	equal the challenge ciphertext $c^\star$ but the encoded public key
	$\tau$ recovered in step (1) of \pdec is such that $\unhide(\tau) =
	X_j^{\star}$ \emph{and} $e_i = e^\star_i$. In this case, the 
	challenger replies with $\bot$ directly, without running 
	$\V_{\mackey}(\cdot)$ (step (5) of \pdec).
\end{game}

\begin{game}[\gamesubmac]
	\label{game:submac}
	As in \gamerepbot, but the challenger replaces the integrity tag in the 
	challenge ciphertext in step (9) of \penc with a random string 
	of the same length.
\end{game}

\begin{game}[\gamesubpay]
	\label{game:subpay}
	As in \gamesubmac, but the challenger replaces the encrypted payload 
	$\ctxtpayload$ in the challenge ciphertext in step (7) of \penc 
	with a random string of the same length.
\end{game}

\para{Conclusion.} As of \gamesubpay, all ciphertexts in the PURBs header, the
payload encryption and the MAC have been replaced by random strings. The open
slots in the hash tables are always filled with random bits. Finally, the
encoded keys $\tau = \hide(X)$ are indistinguishable from random strings as
well, since the keys $X$ are random. Therefore, the PURB ciphertexts $c$ are
indeed indistinguishable from random strings, as in the MSBE game with $b = 1$.

\newcommand{\gameevent}{W}
\begin{proof}

Let $\gameevent_i$ be the event that \A outputs $b^\star = 1$ in game $G_i$.
We aim to show that 
\begin{equation*}
  \begin{split}
  \advmsbe(\onelambda) &= \big| \textrm{Pr}[b^{\star} = 1 \;|\; b = 0 ] -
                                \textrm{Pr}[b^{\star} = 1 \;|\; b = 1 ] \big| \\
                       &= \big| \textrm{Pr}[\gameevent_0] -
                                \textrm{Pr}[\gameevent_7] \big|
  \end{split}
\end{equation*}
is negligible. To do so, we show that each of the steps in the sequence of games
is negligible, i.e., that
$
\big| \textrm{Pr}[\gameevent_i] -
      \textrm{Pr}[\gameevent_{i + 1}] \big|
$
is negligible. The result then follows from the triangle inequality.

\begin{indist}[\gamestart <--> \gamememkeys]
  As long as the KEMs are perfectly correct, the games
  \gamestart and \gamememkeys are identical. Therefore:
  \begin{equation*}
  \big| \textrm{Pr}[\gameevent_0] -
        \textrm{Pr}[\gameevent_1] \big| = 0.
  \end{equation*}
\end{indist}

\begin{indist}[\gamememkeys <--> \gamesubkeys]
	We show that the games \gamememkeys and \gamesubkeys are 
	indistinguishable using a hybrid argument on the number of recipients $r$.
  Consider the hybrid games
	$\gameh{i}$ where the first $i$ recipients
	use random keys $k_1,\ldots, k_i$ as in \gamesubkeys,
	whereas the remaining $r - i$ recipients
	use the real keys $k_{i+1},\ldots, k_r$ as in \gamememkeys.
	Then $\gamememkeys = \gameh{0}$ and $\gamesubkeys = \gameh{r}$.

  We prove that \A cannot distinguish $\gameh{j - 1}$ from
  $\gameh{j}$. Let
	$\suite_j = \langle \G, p, g, \hide(\cdot), \Pi, \hasha, \hashb \rangle$,
  be
  the suite corresponding to recipient $j$.
  Suppose \A can distinguish $\gameh{j - 1}$ from $\gameh{j}$, then we can
  build a distinguisher \B against
  the \ccatwo security of the IES KEM for the suite $\suite'_j = \langle \mathbb{G}, p, 
  g, \hasha \rangle$.
  Recall that \B receives, from 
  its \ccatwo-KEM challenger,
	\begin{compactitem}
		\item a public key $Y$;
		\item a challenge $\langle X^\star, k^{\star} \rangle$, where 
		depending on bit 
		$b \getrand \{0, 1\}$, we have $k^\star = \hasha({Y}^{x^\star})$ if $b = 
		0$ or $k^{\star} \getrand \{0, 1\}^{\lambda_H}$ if $b = 1$ (where
    $\lambda_H$ is the bit-length of $\hasha$);
		\item access to a $\decap(\cdot)$ oracle for all but $X^\star$.
	\end{compactitem}
  At the start of the game, \B will set $pk_j = Y$, so that the public key of
  recipient $j$ matches that of its IES KEM challenger. Note that \B does not
  know the corresponding private key $y_j$. For all other recipients $i$, \B
  sets $(sk_i = y_i,pk_i =  Y_i) = \pkeygen(S_i)$.

  The distinguisher \B will use its challenge $(X^\star, k^{\star})$
  to construct the challenge ciphertext for \A. In particular, when running
  \hdr\encap for a suite $\suite_j$, it sets $X = X^\star$ in step (1) of 
  $\hdr\encap$.
  Moreover, for recipient $j$ it will use $k_j = k^{\star}$. For all other
  recipients $i$ with corresponding suites $\suite_i$ it proceeds as follows 
  when computing $k_i$ in \hdr\encap.
  \begin{compactitem}
  \item If $i < j$, then it sets $k_i \getrand \{0, 1\}^{\lambda_H}$ for
    appropriate $\lambda_H$;
  \item If $i > j$ and the suite $\suite_i$ for user $i$ is the same as suite 
  $\suite_j$ for user $j$, then it sets $k_i = \hasha({X^\star}^{y_i})$; and
  \item If $i > j$, but $S_j \neq S_i$, then it computes $k_i$ as per steps
    (1) and (2) of \hdr\encap.
  \end{compactitem}
  Thereafter, \B continues running \penc as before.
	
  Whenever \B receives a decryption query for a user $pk_i$, it proceeds as
  before. When it receives a decryption query for user $pk_j$, it uses its 
  IES-KEM $\decap$ oracle in step (2) of \hdr\decap. Note that \B is not 
  allowed to call $\decap(\cdot)$ on $X^\star$, but as per the changes in 
  $\gamememkeys$, it will
  directly use $k^\star$ for user $pk_j$ if \hdr\decap recovers $X^\star$ in 
  step (1).

  If $b = 0$ in \B's IES KEM challenge,
  then recipient $j$'s key $k_j = \hasha(Y^{x^\star})$,
  and hence \B perfectly simulates $\gameh{j - 1}$.
  If $b = 1$ in \B's IES KEM challenge,
  then $j$'s key $k_j \getrand \{0, 1\}^{\lambda_H}$
  and, hence, \B perfectly simulates $\gameh{j}$.
  If \A distinguishes $\gameh{j - 1}$ from $\gameh{j}$
  then \B breaks the \ccatwo-KEM security of IES.
  Hence,
  $\gameh{j - 1}$ and $\gameh{j}$ are indistinguishable.
  Repeating this argument $r$ times shows that
  \gamememkeys and \gamesubkeys
  are indistinguishable.
  More precisely:
  \begin{equation*}
  \big| \textrm{Pr}[\gameevent_1] -
        \textrm{Pr}[\gameevent_2] \big| \leq r \cdot \advkemcca(\onelambda).
  \end{equation*}
\end{indist}

\begin{indist}[\gamestart <--> \gamememkeys]
  By perfect correctness of the authentication encryption scheme, we have that
  for all keys $k$ and messages $m$ that $\D_k(\E_k(m)) = m$,
  thus, games \gamesubkeys and \gamememmsg are identical. Therefore:
  \begin{equation*}
  \big| \textrm{Pr}[\gameevent_2] -
        \textrm{Pr}[\gameevent_3] \big| = 0.
  \end{equation*}
\end{indist}

\begin{indist}[\gamememmsg <--> \gamesubenc]
	Similarly to the proof above, consider the hybrid games
	$\gameh{i}$ where the first $i$ entry points are substituted with random 
	strings $e_1, \ldots, e_i$ as in 
	\gamesubenc, 
	whereas the remaining $r - i$ are the actual encryptions as in 
	\gamememmsg.
	Then $\gamememmsg = \gameh{0}$ and $\gamesubenc = \gameh{r}$.
	We show that \A cannot distinguish $\gameh{j - 1}$ from
	$\gameh{j}$. Let 
	$\suite_j = \langle \G, p, g, \hide(\cdot), \Pi, \hasha, \hashb \rangle$,
  be the suite corresponding to recipient $j$. We show that if \A distinguishes
  $\gameh{j - 1}$ from $\gameh{j}$
  then we can build a distinguisher \B against
  the \ccatwo security of $\Pi$.
  \B receives from its \ccatwo challenger:
	\begin{compactitem}
		\item a challenge ciphertext $e^\star$, in response to an encryption call 
		with a message $m$ such that, depending on the bit $b \in \{0, 1\}$, 
		we have that $e^\star = \E_{Z}(m)$ if $b = 0$ or $e^\star$ is a random string if $b 
		= 1$;
		\item a decryption oracle $\D_{Z}(\cdot)$.
	\end{compactitem}

  When constructing the challenge ciphertext, \B calls its challenge oracle with
  $K \parallel \meta$ to obtain $e^\star$, and then sets $e_j^\star = e^\star$ for user $j$'s
  entry point (in line 47 of \layout). We note that in the random oracle the 
  real encryption key $Z_j =
  \hashb(\text{``key''} \parallel k_j)$ is independent from adversary \A's view,
  so we can replace it with the random key of the \ccatwo challenger. For
  other users $i$ it proceeds as follows:
  \begin{compactitem}
  \item If $i < j$, it sets $e_i^\star$ to a random string of appropriate length.
  \item If $i > j$, it computes $e_i^\star$ as per line 47 of \layout.
  \end{compactitem}
  
  Thereafter, \B answers decryption queries as before. Except that whenever, \B
  derives key $k_j$ for user $j$, it will use its decryption oracle 
  $\D_Z(\cdot)$.
  Note that in particular,
  because of the changes in \gamememmsg, 
  \B will not make $\D_Z(\cdot)$ queries on $e^\star_i$ from the challenge 
  ciphertext $c^\star$.

	If $b =0$, \B simulates \gameh{j-1}, and if $b=1$, it simulates 
	\gameh{j}. 
	Therefore, if \A distinguishes between \gameh{j-1} and \gameh{j}, then \B 
	breaks the \ccatwo security of $\Pi$. To show that \gamememmsg is 
	indistinguishable from \gamesubenc, repeat this argument $r$ times.
  More precisely:
  \begin{equation*}
  \big| \textrm{Pr}[\gameevent_3] -
        \textrm{Pr}[\gameevent_4] \big| \leq r \cdot \advae(\onelambda).
  \end{equation*}
\end{indist}

\begin{indist}[\gamesubenc <--> \gamerepbot]
	The challenger's actions in \gamesubenc and \gamerepbot only differ if \A
	could create a decryption request $\qdec(pk_i(\suite_i), c)$ where 
	$\unhide(\tau) = X_i^{\star}$, $e_i = e^\star_i$, and the integrity tag 
	$\sigma$ 
	is valid but $c$ is different from $c^\star$ (recall \A is not allowed to
  query $c^\star$ itself). We show that if \A can cause the challenger to output
  $\bot$ incorrectly, then we can build a simulator \B that breaks the
	strong unforgeability of \mac.
	
	Assume a simulator \B that tries to win an unforgeability game.
	Simulator \B receives access to the oracles $\M(\cdot)$ and $\V(\cdot)$, and needs to output a pair $(c, \sigma)$, such that $\V_{\mackey}(c, \sigma)$ returns true.

  Simulator \B now proceeds as follows. When creating the challenge ciphertext
  $c^\star$, it does not compute $\sigma$ in step (9) of \penc using 
  $K^\star$, but instead uses
  its oracle $\M$ and sets $\sigma = \M(c')$. Note that because of the random
  oracle model for $\hashpayload$ and the fact that \A's view is 
  independent of
  $K^\star$, this change of \mackey remains undetected.

  Whenever \A makes a decryption query $\qdec(pk_i(\suite_i),\allowbreak c)$ \B 
  proceeds as
  before, except when it derives the key $K^*$. In that case it runs $\V(c',
  \sigma)$ to use its oracle to verify the MAC in step (5) of \pdec. If $\V(c',
  \sigma)$ returns $\top$ then \B outputs $(c', \sigma)$ as its forgery (by
  construction, $c'$ was not queried to the MAC oracle $\M(\cdot)$).

  Therefore, \A cannot make queries that cause the challenger to incorrectly
  output $\bot$, and therefore the two games are indistinguishable, provided \mac 
  is strongly unforgeable.
  More precisely:
  \begin{equation*}
  \big| \textrm{Pr}[\gameevent_4] -
        \textrm{Pr}[\gameevent_5] \big| \leq \advmacsforge(\onelambda).
  \end{equation*}
\end{indist}

\begin{indist}[\gamerepbot <--> \gamesubmac]
	If \A can distinguish between \gamerepbot and \gamesubmac, then we 
	can build a distinguisher \B that breaks the indistinguishability from 
	random bits (MAC-IND\$) of \mac.
  
  Distinguisher \B proceeds as follows to compute the challenge ciphertext
  $c^\star$. It proceeds as before, except that in step (9) of \penc, it 
  submits
  $c'$ to its challenge oracle to receive a tag $\tau^\star$. It then sets $\tau
  = \tau^\star$ and proceeds to construct the PURB ciphertext.
  
  Note that as per the changes before, \B never needs to verify a MAC under the
  key that was used to create $\tau^\star$ for the challenge ciphertext.
  Moreover, as before, \A's view is independent of the $K^\star$, so also this
  change of $\mackey$ remains undetected.

	If $b = 0$, \B simulates \gamerepbot, and
  if $b = 1$, \B simulates \gamesubmac.
	Hence, if \A can distinguish between these two games, \B breaks the 
	MAC-IND\$ game.
  More precisely:
  \begin{equation*}
  \big| \textrm{Pr}[\gameevent_5] -
        \textrm{Pr}[\gameevent_6] \big| \leq \advmacind(\onelambda).
  \end{equation*}
\end{indist}

\begin{indist}[\gamesubmac <--> \gamesubpay]
	If \A can distinguish between \gamesubmac and \gamesubpay, then we 
	can build a distinguisher \B that breaks the \cpa property of $(\enc, 
	\dec)$.
	In the \cpa game~\cite{rogaway04nonce}, \B receives:
	\begin{compactitem}
		\item a challenge ciphertext $\ctxtpayload = c_b$, \st $c_0 = 
		\enc_{\enckey}(m)$ on a chosen-by-\B $m$, $c_1 \getrand \{0, 
		1\}^{|c_0|}$, and $b \getrand \{0, 1\}$.
	\end{compactitem}
	\B runs \pdec as before to create a challenge for \A, except that \B uses 
	the \cpa challenge ciphertext $\ctxtpayload$ in step (7), instead of 
	encrypting, as \B does not know $\enckey$.
  As before, \A's view is independent of $K^\star$, so also this
  change of $\enckey$ remains undetected.
	
	\B answers decryption queries $\qdec(pk_i(\suite_i), c)$ from \A as 
  before. In particular
	\begin{compactitem}
		\item if $\unhide(\tau) = X_i^{\star}$ and $e_i = e^\star_i$, \B returns 
		$\bot$ as per the changes in $\gamerepbot$;
		\item Otherwise, \B runs $\pdec(\cdot)$.
	\end{compactitem}
	If $b = 0$, \B simulates \gamesubmac, and, if $b = 1$, \B simulates 
	\gamesubpay.
	Hence, if \A can distinguish between these two games, \B can break the 
	the \cpa property of $(\enc, \dec)$.
  More precisely:
  \begin{equation*}
  \big| \textrm{Pr}[\gameevent_6] -
        \textrm{Pr}[\gameevent_7] \big| \leq \advindcpa(\onelambda).
  \end{equation*}
\end{indist}
Combining the individual inequalities we find that there exists adversaries
$\B_1, \ldots, \B_5$ such that
\begin{align*}
\advmsbe(\onelambda) \leq
  \;& r \left( \advkemcca[\B_1](\onelambda) + \advae[\B_2](\onelambda)\right) + \\
  &\advmacsforge[\B_3](\onelambda) +
  \advmacind[\B_4](\onelambda) + \\
  &\advindcpa[\B_5](\onelambda),
\end{align*}
completing the proof.
\end{proof}

\subsection{Proof of Theorem~\ref{theorem:insider}}
\label{proof:insider}

For our \mspurb \cpa recipient-privacy game, we take inspiration from the 
single-suite recipient-privacy game defined by Barth et 
al.~\cite{barth06privacy}, but we restate it in the \cpa setting.

\begin{game}[Recipient-Privacy]
	\label{game:recpri}
	The game is between a challenger and an adversary \A, and proceeds 
	along the following phases:

  \hangone
	\textbf{Init:} The challenger and adversary take $\lambda$ as input.
  The adversary outputs
  a number of recipients $r$ and corresponding cipher suites
  $\suite_1, \ldots, \suite_r$ it wants to attack. Let $s$ be the number of
  unique cipher suites.
  The challenger verifies,
  for each $i \in \{1, \ldots, r\}$,
  that $\suite_i$ is a valid cipher suite, i.e.,
  that it a valid output of $\msbe\setup(\onelambda)$. The challenger aborts, and
  sets $b^\star \getrand \{0, 1\}$ if the suites are not all valid.
  Adversary \A then outputs two sets of recipients $N_0, N_1 \subseteq \{1, \ldots, 
	n\}$ such that $|N_0| = |N_1| = r$, and the number of users in $N_0$
  and $N_1$ using suite $\suite_j$ is the same.
	
	\hangone
	\textbf{Setup:} For each $i \in 1,\ldots, n$ given by \A, the challenger
	runs $(sk_i, pk_i) \gets \pkeygen(\suite_i)$, where $\suite_i$ is previously 
	chosen by \A.
	The challenger gives two sets $R_0 = \{pk^0_1, \ldots, 
	pk^0_r\}$ and $R_1 = \{pk^1_1, \ldots, pk^1_r\}$ to \A, where $R_0, 
	R_1$ are the generated public keys of the recipients $N_0, N_1$ 
	respectively.
	The challenger also gives to \A all $sk_i$ that correspond to $i \in N_0 
	\cap N_1$.
	
	
	\hangone
	\textbf{Challenge:} \A outputs $m^\star$. The challenger generates 
	$c_0 = \penc(R_0, m^\star)$ and $c_1 = \penc(R_1, m^\star)$.
	The challenger flips a coin $ b \getrand \{0,1\}$ and sends $c^\star = 
	c_b$ to \A.
	
	
	\hangone
	\textbf{Guess:} \A outputs its guess $b^\star$ for $b$ and 
	wins if $b^\star = b$.
\end{game}

We define \A's advantage in this game as:
\begin{equation*}
  \advmsbein(\onelambda) = 2\left| \textrm{Pr}[b = b^{\star}] - \tfrac{1}{2} \right|.
\end{equation*}
We say that a MSBE scheme is \texttt{cpa}-secure against insiders if $\advmsbein(\onelambda)$ is negligible in
the security parameter.

The conditions on $N_0$ and $N_1$ in the game ensure that \A cannot trivially win by
looking at the size of the ciphertext. PURBs allows for suites with different
groups (resulting in different size encodings of the corresponding IES public
key) and for suites to use different authenticated encryption schemes
(that could result in different sizes of encrypted entry points). Since PURBs
must encode groups and entry points into the header, we mandate that for each
suite the number of recipients is the same in $N_0$ and $N_1$.
This assumption is similar to requiring equal-size sets of recipients in a 
challenge game for single-suite broadcast encryption~\cite{barth06privacy}.
As in broadcast encryption, if this requirement is an issue, a sender can 
add dummy recipients to avoid structural leakage to an insider adversary.

We will show that
\begin{equation*}
  \advmsbein(\onelambda) \leq 2 d \cdot \advkemcca[\B](\onelambda),
\end{equation*}
where $d$ is the number of recipients in which $N_0$ and $N_1$ differ.

\begin{proof}
Similarly to Barth et al.~\cite{barth06privacy}, we prove recipient privacy 
when the sets $R_0$ and $R_1$ differ only by one public key in one suite.
The general case follows by a hybrid argument.
Consider the following games:

\begin{game}[\gameinstart]
	\label{game:instart}
	This game is as the original recipient-privacy \cpa game where $b = 0$ 
	and $pk_i = R_0 \setminus R_1$, $pk_j = R_1 \setminus R_0$, where
	the public keys $pk_i$ and $pk_j$ are of the same suite $\suite$.
\end{game}

\begin{game}[\gamerandkey]
	As in \gameinstart, but we change how a key $k_i^\star$ 
	corresponding to the recipient $i$ is computed in $\hdr\encap$ for
	the \emph{challenge ciphertext}.
	Instead of computing $k_i^\star = \hasha(Y_i^{x})$ (where $Y_i = pk_i$) 
	as in step (2) of $\hdr\encap$, we set $k_i^\star \getrand \{0, 
	1\}^{\lambda_H}$.
	As the challenger generates fresh public keys for each encryption query 
	and thus a fresh key $k_i$, and does not have to answer decryption queries,
  it does not need to memorize $k_i^\star$.
\end{game}

\begin{game}[\gameinend]
	As in \gamerandkey, but we change the random sampling $k_i^\star$ in 
	$\hdr\encap$ for the challenge ciphertext with $k_i^\star = 
	\hasha(Y_j^{x}) = k_j^\star$ where $Y_j = pk_j$.
	The game now is the original recipient-privacy \cpa game where $b = 1$.
\end{game}

\para{Conclusion.}
\gameinstart represents the recipient-privacy game with $b=0$ and 
\gameinend recipient-privacy game with $b=1$.
If \A cannot distinguish between \gameinstart and \gameinend, \A does not 
have an advantage in winning the recipient-privacy game.

Let $\gameevent_i$ be the event that \A outputs $b^\star = 1$ in game $G_i$.

\begin{indist}[\gameinstart <--> \gamerandkey]
	If \A can distinguish between \gameinstart and \gamerandkey, we can 
	build a distinguisher \B against the \ccatwo security of the IES KEM.
	Recall that \B receives, from its \ccatwo-KEM challenger,
	\begin{compactitem}
		\item a public key $Y$;
		\item a challenge $\langle X^\star, k^{\star} \rangle$, where 
		depending on bit $b \getrand \{0, 1\}$, we have $k^\star = 
		\hasha({Y}^{x^\star})$ if $b = 0$ or $k^{\star} \getrand \{0, 
		1\}^{l(\lambda)}$ if $b = 1$;
		\item access to a $\decap(\cdot)$ oracle for all but $X^\star$.
	\end{compactitem}
	At the start of the game, \B will set $pk_i = Y$, so that the public key of
	recipient $i$ matches that of its IES KEM challenger. Note that \B does 
	not know the corresponding private key $y_i$. For all other recipients 
	$h$, \B sets $(sk_h = y_h, pk_h =  Y_h) = \pkeygen(S_h)$.
	As \A plays an \cpa game, \B does not need to use the $\decap(\cdot)$ 
	oracle (in fact, for \cpa recipient privacy \cpa security of the IES KEM
  suffices).
	
	If $b=0$ in the IES-KEM challenge, then \B simulates \gameinstart, and, 
	If $b=1$, \B simulates \gamerandkey. Hence, if \A distinguishes between 
	\gameinstart and \gamerandkey, \B wins in the \ccatwo IES-KEM game.
  Therefore:
  \begin{equation*}
    \left| \textrm{Pr}[\gameevent_0] -
    \textrm{Pr}[\gameevent_1] \right| \leq \advkemcca[\B](\onelambda)
  \end{equation*}
  
\end{indist}

\begin{indist}[\gamerandkey <--> \gameinend]
	The proof follows the same steps as the proof of 
	\gameinstart <--> \gamerandkey.
  Therefore:
  \begin{equation*}
    \left| \textrm{Pr}[\gameevent_0] -
    \textrm{Pr}[\gameevent_1] \right| \leq \advkemcca[\B](\onelambda).
  \end{equation*}
\end{indist}

Let $d$ be the number of recipients that differ in $N_0$ and $N_1$.
Then by repeating the above two steps $d$ times in a hybrid argument, we find
that:
\begin{equation*}
  \advmsbein(\onelambda) \leq 2 d \cdot \advkemcca[\B](\onelambda),
\end{equation*}
as desired.
\end{proof}

\end{document}